%% file: arxiv-main.tex
\documentclass[a4paper,english,11pt]{article}
\usepackage{lmodern}        

\usepackage[T1]{fontenc}
\usepackage[latin9]{inputenc}
\usepackage[top=1in, bottom=1in, left=1in, right=1in]{geometry}
\usepackage{titlesec}
\usepackage{verbatim}
\usepackage{float}
\usepackage{amssymb,amsfonts,amsmath}
\usepackage{enumerate}
\usepackage{enumitem}
\usepackage{amsthm}

\usepackage{algorithm}
\usepackage{algpseudocode}

\usepackage{color}
\usepackage{graphicx}

\usepackage{amsthm}

\usepackage{hyperref}
\usepackage{caption}
\usepackage{subfigure}

\floatstyle{ruled}
\newfloat{Figure}{tbp}{loa}
\floatname{Figure}{Figure}

\usepackage{microtype}

\DeclareGraphicsExtensions{.eps}

\newcommand\nc\newcommand
\nc\bfa{{\boldsymbol a}}\nc\bfA{{\boldsymbol A}}\nc\cA{{\mathcal A}}
\nc\bfb{{\boldsymbol b}}\nc\bfB{{\boldsymbol B}}\nc\cB{{\mathcal B}}
\nc\bfc{{\boldsymbol c}}\nc\bfC{{\boldsymbol C}}\nc\cC{{\mathcal C}}
\nc\sC{{\mathscr C}}
\nc\bfd{{\boldsymbol d}}\nc\bfD{{\boldsymbol D}}\nc\cD{{\mathcal D}}
\nc\bfe{{\boldsymbol e}}\nc\bfE{{\boldsymbol E}}\nc\cE{{\mathcal E}}
\nc\bff{{\boldsymbol f}}\nc\bfF{{\boldsymbol F}}\nc\cF{{\mathcal F}}
\nc\bfg{{\boldsymbol g}}\nc\bfG{{\boldsymbol G}}\nc\cG{{\mathcal G}}
\nc\bfh{{\boldsymbol h}}\nc\bfH{{\boldsymbol H}}\nc\cH{{\mathcal H}}
\nc\bfi{{\boldsymbol i}}\nc\bfI{{\boldsymbol I}}\nc\cI{{\mathcal I}}
\nc\bfj{{\boldsymbol j}}\nc\bfJ{{\boldsymbol J}}\nc\cJ{{\mathcal J}}
\nc\bfk{{\boldsymbol k}}\nc\bfK{{\boldsymbol K}}\nc\cK{{\mathcal K}}
\nc\bfl{{\boldsymbol l}}\nc\bfL{{\boldsymbol L}}\nc\cL{{\mathcal L}}
\nc\bfm{{\boldsymbol m}}\nc\bfM{{\boldsymbol M}}\nc\sM{{\mathscr M}}\nc\cM{{\mathcal M}}
\nc\bfn{{\boldsymbol n}}\nc\bfN{{\boldsymbol N}}\nc\cN{{\mathcal N}}
\nc\bfo{{\boldsymbol o}}\nc\bfO{{\boldsymbol O}}\nc\cO{{\mathcal O}}
\nc\bfp{{\boldsymbol p}}\nc\bfP{{\boldsymbol P}}\nc\cP{{\mathcal P}}
\nc\bfq{{\boldsymbol q}}\nc\bfQ{{\boldsymbol Q}}\nc\cQ{{\mathcal Q}}
\nc\bfr{{\boldsymbol r}}\nc\bfR{{\boldsymbol R}}\nc\cR{{\mathcal R}}
\nc\bfs{{\boldsymbol s}}\nc\bfS{{\boldsymbol S}}\nc\cS{{\mathcal S}}
\nc\bft{{\boldsymbol t}}\nc\bfT{{\boldsymbol T}}\nc\cT{{\mathcal T}}
\nc\bfu{{\boldsymbol u}}\nc\bfU{{\boldsymbol U}}\nc\cU{{\mathcal U}}
\nc\bfv{{\boldsymbol v}}\nc\bfV{{\boldsymbol V}}\nc\cV{{\mathcal V}}
\nc\bfw{{\boldsymbol w}}\nc\bfW{{\boldsymbol W}}\nc\cW{{\mathcal W}}
\nc\bfx{{\boldsymbol x}}\nc\bfX{{\boldsymbol X}}\nc\cX{{\mathcal X}}
\nc\bfy{{\boldsymbol y}}\nc\bfY{{\boldsymbol Y}}\nc\cY{{\mathcal Y}}
\nc\bfz{{\boldsymbol z}}\nc\bfZ{{\boldsymbol Z}}\nc\cZ{{\mathcal Z}}

\nc\diff{{\mathrm d}}
\nc\e{{\mathrm e}}
\nc\calC{{\mathcal C}}

\DeclareMathOperator{\average}{avg}

\newcommand{\remove}[1]{}

\newcommand{\avg}{{\mathbb E}}

\newtheorem{lemma}{Lemma}
\newtheorem{corollary}{Corollary}
\newtheorem{theorem}{Theorem}
\theoremstyle{definition}
\newtheorem*{definition}{Definition}
\theoremstyle{corollaryn}
\newtheorem*{corollaryn}{Corollary}
\newtheorem*{problem}{Problem}
\newtheorem{example}{Example}

\newtheorem{claim}{Claim}
\theoremstyle{theorem-n}
\newtheorem*{theorem-n}{Theorem}

\newcommand{\cc}{{\sf Crowd-Cluster}}

\def\DEBUG{true}

\ifdefined\DEBUG

  \def\rem#1{{\marginpar{\raggedright\scriptsize #1}}}
  \newcommand{\barnr}[1]{\rem{\textcolor{red}{$\bullet$ #1}}}
  \newcommand{\aryar}[1]{\rem{\textcolor{green}{$\bullet$ #1}}}
\else

  \newcommand{\barnr}[1]{}
  \newcommand{\aryar}[1]{}

\fi

\newcommand\reals{{\mathbb R}}

\allowdisplaybreaks

\author{Arya Mazumdar\thanks{University of Massachusetts Amherst, \url{arya@cs.umass.edu}. This work is supported in part by an NSF CAREER award CCF 1453121 and NSF award CCF 1526763.}~~  ~~~~ Barna Saha\thanks{University of Massachusetts Amherst, \url{barna@cs.umass.edu}.This work is partially supported by a NSF CCF 1464310 grant, a Yahoo ACE Award and a Google Faculty Research Award.} \\
College of Information \& Computer Science\\
University of Massachusetts Amherst\\
Amherst, MA, 01002
}

\title{Clustering via Crowdsourcing}
\date{}
\begin{document}
\sloppy
\maketitle

\pagestyle{plain}
\setcounter{page}{0}
\begin{abstract}
In recent years, crowdsourcing, aka human aided computation has emerged as an effective platform for solving problems that are considered complex for machines alone. Using human is time-consuming and costly due to monetary compensations. Therefore, a crowd based algorithm must judiciously use any information computed through an automated process, and ask minimum number of questions to the crowd adaptively. 

One such problem which has received significant attention is {\em entity resolution}. Formally, we are given a graph $G=(V,E)$ with unknown edge set $E$ where $G$ is a union of $k$ (again unknown, but typically large $O(n^\alpha)$, for $\alpha>0$) disjoint cliques $G_i(V_i, E_i)$, $i =1, \dots, k$. The goal is to retrieve the sets $V_i$s by making minimum number of pair-wise queries $V \times V\to\{\pm1\}$ to an oracle (the crowd). When the answer to each query is correct, e.g. via resampling, then this reduces to finding connected components in a graph. On the other hand, when crowd answers may be incorrect, it corresponds to clustering over minimum number of noisy inputs. Even, with perfect answers, a simple lower and upper bound of $\Theta(nk)$ on query complexity can be shown. A major contribution of this paper is to reduce the query complexity to linear or even sublinear in $n$ when mild side information is provided by a machine, and even in presence of crowd errors which are not correctable via resampling. We develop new information theoretic lower bounds on the query complexity of clustering with side information and errors, and our upper bounds closely match with them. Our algorithms are naturally parallelizable, and also give near-optimal bounds on the number of adaptive rounds required to match the query complexity.

\end{abstract}
\newpage
\section{Introduction}
Consider we have an undirected graph $G(V\equiv[n],E)$, $[n]\equiv\{1,\dots, n\}$, such that $G$ is a union of $k$ disjoint cliques $G_i(V_i, E_i)$, $i =1, \dots, k$, but the subsets $V_i \subset [n]$, $k$ and $E$ are unknown to us. We want to make minimum number of adaptive pair-wise queries from $V \times V$ to an oracle, and recover the clusters. Suppose, in addition, we are also given a noisy weighted similarity matrix $W=\{w_{i,j}\}$ of $G$, where $w_{i,j}$ is drawn from a probability distribution $f_+$ if $i$ and $j$ belong to the same cluster, and else from $f_{-}$. However, the algorithm designer does not know either $f_{+}$ or $f_{-}$. How does having this side information affect the number of queries to recover the clusters, which in this scenario are the hidden connected components of $G$? To add to it, let us also consider the case when some of the answers to the queries are erroneous. We want to recover the  clusters with minimum number of noisy inputs possibly with the help of some side information. In the applications that motivate this problem, the oracle is the crowd. 

In the last few years, crowdsourcing has emerged as an effective solution for large-scale ``micro-tasks''. Usually, the micro-tasks that are accomplished using crowdsourcing tend to be those that computers cannot solve very effectively, but are fairly trivial for humans with no specialized training. Consider for example six places, all named after John F. Kennedy \footnote{
\url{http://en.wikipedia.org/wiki/Memorials_to_John_F._Kennedy}}:
 ($r_a$)~{\sf John F. Kennedy International Airport},
($r_b$)~{\sf JFK Airport}, 
($r_c$)~{\sf Kennedy Airport, NY}
($r_d$)~{\sf John F. Kennedy Memorial Airport},
($r_e$)~{\sf Kennedy Memorial Airport, WI},
($r_f$)~{\sf John F. Kennedy Memorial Plaza}. 
Humans can determine using domain knowledge that the above six places correspond to 
three different entities: $r_a, r_b$, and $r_c$ refer to one entity, $r_d$ and 
$r_e$ refer to a second entity, and $r_f$ refers to a third entity. However, for a computer, it is hard to distinguish them. This problem known as entity resolution is a basic task in classification, data mining and database management \cite{fellegi1969theory,elmagarmid2007duplicate,getoor2012entity}. It has many alias in literature, and also known as coreference/identity/name/record resolution, entity disambiguation/linking, duplicate detection, deduplication, record matching etc. There are several books that just focus on this topic \cite{christen2012data,Herzog:2007}. For a comprehensive study and applications, see \cite{getoor2012entity}.

Starting with the work of Marcus et al.~\cite{marcus2011human}, there has been a flurry of works that have aimed at using human power for entity resolution \cite{gokhale2014corleone,vesdapunt2014crowdsourcing,dkmr:14,wang2012crowder,fss:16,DBLP:conf/icde/VerroiosG15, dalvi2013aggregating,ghosh2011moderates,
karger2011iterative}. Experimental results using crowdsourcing platforms such as Amazon Mechanical Turk have exceeded the machine only performance \cite{vesdapunt2014crowdsourcing,wang2012crowder}. In all of these works, some computer generated pair-wise similarity matrix is used to order the questions to crowd. Using human in large scale experiments is costly due to monetary compensation paid to them, in addition to being time consuming. Therefore, naturally these works either implicitly or explicitly aim to minimize the number of queries to crowd. 
%
Assuming the crowd returns answers correctly, entity resolution using crowdsourcing corresponds exactly to the task of finding connected components of $G$ with minimum number of adaptive queries to $V \times V$. Typically $k$ is large \cite{vesdapunt2014crowdsourcing,wang2012crowder,fss:16,DBLP:conf/icde/VerroiosG15}, and we can, not necessarily,  take $k \geq n^{\alpha}$ for some constant $\alpha \in [0,1]$.

It is straightforward to obtain an upper bound of $nk$ on the number of queries: simply ask one question per cluster for each vertex, and is achievable even when $k$ is unknown. Except for this observation \cite{vesdapunt2014crowdsourcing,dkmr:14}, no other theoretical guarantees on the query complexity were known so far. Unfortunately, $\Omega(nk)$ is also a lower bound \cite{dkmr:14}. Bounding query complexity of basic problems like selection and sorting have received significant attention in the theoretical computer science community \cite{frpu:94,bmw:16,akss:86, bg:90}. Finding connected components is the most fundamental graph problem, and given the matching upper and lower bounds, there seems to be a roadblock in improving its query complexity beyond $nk$.

In contrast, the heuristics developed in practice often perform much better, and all of them use some computer generated similarity matrix to guide them in selecting the next question to ask. We call this {\em crowdsourcing using side information}. So, we are given a similarity matrix $W=\{w_{i,j}\}_{i,j \in V\times V}$, which is a noisy version of the original adjacency matrix of $G$ as discussed in the beginning. Many problems such as sorting, selection, rank aggregation etc. have been studied using noisy input where noise is drawn from a distribution \cite{bm:08,bm:09,mmv:13}. Many probabilistic generative models, such as stochastic block model, are known for clustering \cite{abh:16,hajek2015achieving,chin2015stochastic,mossel2014consistency}. However, all of these works assume the underlying distributions are known and use that information to design algorithms.  Moreover, none of them consider query complexity while dealing with noisy input. 

 We show that with side information, even with unknown $f_+$ and $f_-$, a drastic reduction in query complexity is possible. We propose a randomized algorithm that reduces the number of queries from $O(nk)$ to $\tilde{O}(\frac{k^2}{\Delta(f_+,f_-)})$, where $\Delta(f_+,f_-)\equiv D(f_+\|f_-)+D(f_-\|f_+)$ and $D(p\|q)$ is the Kullback-Leibler divergence between the probability distributions $p$, and $q$, and recovers the clusters accurately with high probability. Interestingly, we show $\Omega\left(\frac{k^2}{\Delta(f_+,f_-)}\right)$ is also an information-theoretic lower bound, thus matching the query complexity upper bound within a logarithmic factor. This lower bound could be of independent interest, and may lead to other lower bounds in related communication complexity models. To obtain the clusters accurately with probability $1$, we propose a Las Vegas algorithm with expected query complexity $\tilde{O}\left(n+\frac{k^2}{\Delta(f_+,f_-)}\right)$ which again matches the corresponding lower bound. 
 
 So far, we have considered the case when crowd answers are accurate. It is possible that crowd answers contain errors, and remain erroneous even after repeating a question multiple times. That is, resampling, repeatedly asking the same question and taking the majority vote, does not help much. Such observation has been reported in \cite{DBLP:conf/icde/VerroiosG15,DBLP:journals/corr/GruenheidNKGK15} where resampling only reduced errors by $\sim 20\%$. Crowd workers often use the same source (e.g., Google) to answer questions. Therefore, if the source is not authentic, many workers may give the same wrong answer to a single question. Suppose that error probability is $p <\frac{1}{2}$. Under such crowd error model, our problem becomes that of clustering with noisy input, where this noisy input itself is obtained via  adaptively querying the crowd. 
 
  We give the first information theoretic lower bounds in this model to obtain the maximum likelihood estimator, and again provide nearly matching upper bounds with and without side information. Side information helps us to drastically reduce the query complexity, from $\tilde{O}(\frac{nk}{D(p\|(1-p))})$  to  
 $\tilde{O}(\frac{k^2}{D(p\|1-p)\Delta(f_+,f_-)})$ where $D(p\|1-p) = (1-2p)\log\frac{1-p}{p}$.  An intriguing fact about this algorithm is that it has running time $O(k^{\frac{\log{n}}{D(p\|1-p)}})$, and assuming the conjectured hardness of finding planted clique from an Erd\H{o}s-R\'{e}nyi random graph \cite{hk:11}, this running time cannot be improved\footnote{Note that a query complexity bound does not necessarily removes the possibility of a super-polynomial running time.}. However, if we are willing to pay a bit more on the query complexity, then the running time can be made into polynomial. This also provides a better bound on an oft-studied clustering problem, correlation clustering over noisy input \cite{ms:10,bbc:04}. While prior works have considered sorting without resampling \cite{bm:09}, these are the first results to consider crowd errors for a natural clustering problem.

The algorithms proposed in this work are all intuitive, easily implementable, and can be parallelized. 
They do not assume any knowledge on the value of $k$, or the underlying distributions $f_+$ and $f_-$. On the otherhand, our information theoretic lower bounds work even with the complete knowledge of $k, f_+,f_-$. While queries to crowd can be made adaptively, it is also important to minimize the number of adaptive rounds required maintaining the query upper bound. Low round complexity helps to obtain results faster. We show that all our algorithms extend nicely to obtain close to optimal round complexity as well. Recently such results have been obtained for sorting (without any side information) \cite{bmw:16}. Our work extends nicely to two more fundamental problems: finding connected components, and noisy clustering.


\subsection{Related Work}
In a recent work \cite{bmw:16}, Braverman, Mao and Weinberg studied the round complexity of selection and obtaining the top-$k$ and bottom-$k$ elements when crowd answers are all correct, or are erroneous with probability $\frac{1}{2}-\frac{\lambda}{2}$, or erased with probability $1-\lambda$, for some $\lambda >0$. They do not consider any side information. There is an extensive literature of algorithms in the TCS community where the goal is to do either selection or sorting with $O(n)$ comparisons in the fewest interactive rounds, aka parallel algorithms for sorting \cite{v:75,r:81,aa:87,aa:88,akss:86, bg:90}. However, those works do not consider any erroneous comparisons, and of course do not incorporate side information. Feige et al., study the depth of noisy decision tree for simple boolean functions, and selection, sorting, ranking etc. \cite{frpu:94}, but not with any side information. Parallel algorithms for finding connected components and clustering have similarly received a huge deal of attention \cite{h:91,g:94,ckt:96,pporrj:15}. Neither those works, nor their modern map-reduce counterparts \cite{ksv:10,eim:11,gsz:11,ahn:15} study query complexity, or noisy input. There is an active body of work dealing with sorting and rank aggregation with noisy input under various models of noise generation \cite{bm:08,bm:09,mmv:13}. However these works aim to recover the maximum likelihood ordering without any querying. Similarly, clustering algorithms like correlation clustering has been studied under various random and semirandom noise models without any active querying \cite{bbc:04,ms:10,mmv:14}. Stochastic block model is another such noisy model which has recently received a great deal of attention \cite{abh:16,hajek2015achieving,chin2015stochastic,mossel2014consistency}, but again prior to this, no work has considered the querying capability when dealing with noisy input. In all these works, the noise model is known to the algorithm designer, since otherwise the problems become NP-Hard \cite{bm:08,bbc:04,acn:08}.

In more applied domains, 
many frameworks have been developed to leverage humans for performing entity resolution 
~\cite{wang2012crowder,gokhale2014corleone}.
Wang et al.~\cite{wang2012crowder} describe a hybrid human-machine
framework CrowdER, that automatically detects pairs or clusters that 
have a high likelihood of matching based on a similarity function, which are then verified by humans. Use of similarity function is common across all these works to obtain querying strategies~\cite{gokhale2014corleone,vesdapunt2014crowdsourcing}, but hardly any provide bounds on the query complexity. The only exceptions are \cite{vesdapunt2014crowdsourcing, dkmr:14} where a simple $nk$ bound on the query complexity has been derived when crowd returns correct answers, and no side information is available. This is also a lower bound even for randomized algorithms \cite{dkmr:14}. Firmani et al. \cite{fss:16} analyzed the algorithms of \cite{wang2012crowder} and \cite{vesdapunt2014crowdsourcing} under a very stringent noise model.

To deal with the possibility that the crowdsourced oracle may 
give wrong answers, there are simple majority voting mechanisms or more
complicated heuristic techniques~\cite{DBLP:conf/icde/VerroiosG15, dalvi2013aggregating,ghosh2011moderates,
karger2011iterative} to handle such errors. No theoretical guarantees exist in any of these works. Davidson et al., consider a variable error model where clustering is based on a numerical value--in that case clusters are intervals with few jumps (errors), and the queries are unary (ask for value) \cite{dkmr:14}. This error model is not relevant for pair-wise comparison queries.

\subsection{Results and Techniques}

\begin{problem}[\cc]

Consider an undirected graph $G(V\equiv[n],E)$, such that $G$ is a union of $k$ disjoint cliques (clusters) $G_i(V_i, E_i)$, $i =1, \dots, k$, where $k$, the subsets $V_i \subseteq [n]$ and $E$ are unknown. There is an oracle $\mathcal{O}: V\times V \to \{\pm1\},$ which takes as input a pair of vertices $u,v \in V \times V$, and returns either $+1$ or $-1$. Let $\mathcal{O}(Q)$, $Q \subseteq V \times V$ correspond to oracle answers to all pairwise queries in $Q$. The queries in $Q$ can be done adaptively.

The adjacency matrix of $G$ is a block-diagonal matrix. Let us denote this matrix by $A = (a_{i,j})$.
 Consider $W$,  an $n \times n$ matrix, which is the noisy version of the matrix $A$.
Assume that the $(u,v)$th entry of the matrix $W$, $w_{u,v}$, is a nonnegative random variable in $[0,1]$ drawn from a probability density or mass function $f_+$ for $a_{i,j} =1$, and is drawn from a probability density or mass function $f_{-}$ if $a_{i,j} =0$. $f_{+}$ and $f_{-}$ are unknown.
\begin{itemize}
\item {\bf \cc~with Perfect Oracle} Here $\mathcal{O}(u,v)=+1$ iff $u$ and $v$ belong to the same cluster and $\mathcal{O}(u,v)=-1$ iff $u$ and $v$ belong to different clusters.
\begin{enumerate}
\item {\bf Without Side Information.} Given $V$, find $Q \subseteq V \times V$ such that $|Q|$ is minimum, and from  $\mathcal{O}(Q)$ it is possible to recover $V_i$, $i=1,2,...,k$.
\item {\bf With Side Information.} Given $V$ and $W$, find $Q \subseteq V \times V$ such that $|Q|$ is minimum, and from  $\mathcal{O}(Q)$ it is possible to recover $V_i$, $i=1,2,...,k$.
\end{enumerate}
\item {\bf \cc~with Faulty Oracle} There is an error parameter $p=\frac{1}{2}-\lambda$ for some $\lambda >0$. 
We denote this oracle by  $\mathcal{O}_p$. Here if $u,v$ belong to the same cluster then $\mathcal{O}_p(u,v)=+1$ with probability $1-p$ and $\mathcal{O}_p(u,v)=-1$ with probability $p$. On the otherhand, if $u,v$ do not belong to the same cluster then $\mathcal{O}_p(u,v)=-1$ with probability $1-p$ and $\mathcal{O}_p(u,v)=+1$ with probability $p$ (in information theory literature, such oracle is called {\em binary symmetric channel}). 
\begin{enumerate}
\item {\bf Without Side Information.} Given $V$, find $Q \subseteq V \times V$ such that $|Q|$ is minimum, and from  $\mathcal{O}_{p}(Q)$ it is possible to recover $V_i$, $i=1,2,...,k$ with high probability.
\item {\bf With Side Information.} Given $V$ and $W$, find $Q \subseteq V \times V$ such that $|Q|$ is minimum, and from  $\mathcal{O}_{p}(Q)$ it is possible to recover $V_i$, $i=1,2,...,k$ with high probability.
\end{enumerate}
\item {\bf \cc~with Round Complexity} Consider all the above problems where $\mathcal{O}$ (similarly $\mathcal{O}_{p}$) can answer to $n\log{n}$ queries simultaneously, and the goal is to minimize the number of adaptive rounds of queries required to recover the clusters.
\end{itemize}
\end{problem}
\subsubsection{Lower Bounds}
When no side information is available, it is somewhat straight-forward to have a lower bound on the query complexity if the oracle is perfect. Indeed, in that case the query complexity of \cc~is $\Omega(nk)$ where $n$ is the total number of elements and $k$ is the number of clusters.

To see this, note that, any algorithm can be provided with a clustering designed adversarially in the following way. First, $k$ elements residing in $k$ different clusters are revealed to the 
algorithm. For a vertex among the remaining $n-k$ vertices, if the algorithm makes any less than $k-2$ queries, the adversary still can place the vertex in one of the remaining $2$ clusters--resulting in a query complexity of $(n-k)(k-1)$. 
This argument can be extended towards randomized 
algorithms as well, by using Yao's min-max principal, and has been done in \cite{dkmr:14}.  However \cite{dkmr:14} left open the case of proving lower bound for randomized algorithms when the clusters are nearly balanced (ratio between the minimum and maximum cluster size is bounded). One of the lower bound results proved in this paper resolves it. 

Our main technical results for perfect oracle are for \cc~with side information. Our lower bound results are information theoretic, and can be summarized in the following theorem.

\begin{theorem}\label{thm:lb-main} 
Any (possibly randomized) algorithm with the knowledge of $f_+, f_-,$ and the number of clusters $k$,  that does not perform at least $\Omega\Big(\frac{k^2}{\Delta(f_+,f_-)}\Big)$ queries, $\Delta(f_+,f_-)> 0$, will be unable to return the correct clustering  with probability at least $\frac1{10}$. (Proof in Sec.~\ref{sec:perfectlb}).
\end{theorem}
\begin{corollary}
\label{cor:lb-main}
Any (possibly randomized but Las Vegas) algorithm with the knowledge of $f_+, f_-,$ and the number of clusters $k$,  that does not perform at least $\Omega\Big(n+\frac{k^2}{\min\{1,\Delta(f_+,f_-)\}}\Big)$ queries, $\Delta(f_+,f_-)> 0$, will be unable to return the correct clustering. (Proof in Sec.~\ref{sec:perfectlb}).
\end{corollary}
The main high-level technique is the following. Suppose, a vertex is to be assigned to a cluster. We have some side-information and answers to queries involving this vertex at hand. Let these constitute a random variable 
$X$ that we have observed. Assuming that there are $k$ possible clusters to assign this vertex to, we have a $k$-hypothesis testing problem. 
By observing $X$, we have to decide which of the $k$ different distributions (corresponding to the vertex being in $k$ different clusters)
it is coming from. If the distributions are very close  (in the sense of total variation distance or divergence), then we are bound to make an error in deciding.

We can compare this problem of assigning a vertex to one of the $k$-clusters to finding a biased coin among $k$ coins. In the later problem, we are asked to find out the minimum number of coin tosses needed for correct identification. This type of idea has previously been applied
 to design adversarial strategies that lead to lower bounds on average regret for the multi-arm bandit problem (see, \cite{auer2002nonstochastic, cesa2006prediction}).

The problem that we have in hand, for lower bound on query-complexity, is substantially different.  It becomes  a nontrivial task to identify the correct input and design 
the set-up so that we can handle the problem in the framework of finding a biased coin. The key insight here is that, given a vertex, the combined side-information pertaining to this vertex and a cluster plays the role of tossing a particular coin (multiple times) in the coin-finding problem. However the liberty of an algorithm designer to query freely 
creates the main  challenge.



For faulty oracle, note that we are not allowed to ask the same question multiple times to get the correct answer with high probability. This changes the situation quite a bit, though in some sense this is closer to coin-tossing experiment than the previous one as we handle binary random variables here (the answer to the queries). 
We first note that, for faulty-oracle, even for probabilistic recovery a minimum size bound on cluster size is required. 
For example, consider the following two different clusterings.
$C_1: V = \sqcup_{i=1}^{k-2} V_i \sqcup\{v_1,v_2\}\sqcup\{v_3\}$ and $C_2: V =   \sqcup_{i=1}^{k-2} V_i \sqcup\{v_1\}\sqcup\{v_2,v_3\}$. Now if one of these two clusterings are given two us uniformly at random, no matter how many queries we do, we will fail to recover the correct cluster with probability at least $p$. Our lower bound result works even when all the clusters are close to their average size (which is $\frac{n}{k}$), and resolves a question from \cite{dkmr:14} for $p=0$ case.

This removes  the constraint on the algorithm designer on how many times a cluster can be queried with a vertex and the algorithms can have greater flexibility.
While we have to show that enough number of queries must be made with a large number of vertices $V' \subset V$, either of the conditions on minimum or maximum sizes of a cluster ensures that $V'$ contains enough vertices that do not satisfy this query requirement.

\begin{theorem} \label{thm:faulty}
Assume either of the following cases:
\begin{itemize}
\item  the maximum size of a cluster is $\le \frac{4n}{k}$.
\item   the minimum size of a cluster is $\ge \frac{n}{20k}$.
\end{itemize}
For a clustering that satisfies either of the above two conditions, any (randomized) algorithm must make $\Omega\Big(\frac{nk}{D(p\|1-p)}\Big)$ queries
to recover the correct clusters with  probability $0.9$ when $p >0$. For $p =0$ any (randomized) algorithm must make $\Omega(nk)$ queries
to recover the correct clusters with  probability $0.9$. (Proof in Sec.~\ref{sec:error-lc}).
\end{theorem}

We believe that our lower bound techniques are of independent interest, and can spur new lower bounds for communication complexity problems.

\subsubsection{Upper Bounds}

Our upper bound results are inspired by the lower bounds. For \cc~with perfect oracle, a straight forward algorithm achieves a $nk$ query complexity. One of our main contributions is a drastic reduction in query complexity of \cc~when side information is provided. Let $\mu_+\equiv \int{xf_+(x)dx}, \mu_-\equiv\int{xf_-(x)dx}$. Our first theorem that assumes $\mu_+>\mu_-$ is as follows.

\begin{theorem}[Perfect Oracle+Side Information]
\label{thm:mu}
With known $\mu_+,\mu_-$, there exist a Monte Carlo algorithm for \cc~with query complexity $O(\frac{k^2\log{n}}{(\mu_+-\mu_{-})^2})$,  and a Las Vegas algorithm with expected query complexity $O(n+\frac{k^2\log{n}}{(\mu_+-\mu_{-})^2})$ even when $\mu_+,\mu_-$ are unknown. (Proof in Sec.~\ref{sec:perfectub}).
\end{theorem}

Many natural distributions such as $\mathcal{N}(\mu_+,1)$ and $\mathcal{N}(\mu_-,1)$ have  $\Delta(\mathcal{N}(\mu_+,1)\|\mathcal{N}(\mu_-,1))=(\mu_+-\mu_{-})^2$. But, it is also natural to have distributions where $\mu_+=\mu_-$ 
but $\Delta(f_+,f_-) > 0$. As a simple example, consider two discrete distributions with mass $1/4,1/2,1/4$ and $1/3,1/3,1/3$ respectively at points $0,1/2,1$. Their means are the same, but divergence is  constant.The following theorem matches the lower bound upto a $\log{n}$ factor with no assumption on $\mu_+,\mu_-$.  
 \begin{theorem}[Perfect Oracle+Side Information]
 \label{thm:div-new}
 Let $f_+$ and $f_-$ be pmfs \footnote{We can handle probability density functions as well for Theorem \ref{thm:div-new} and Theorem \ref{thm:div-new-err}, if the quantization error is small. Our other theorems are valid for $f_+$ and $f_-$ being both probability mass functions and density functions.} and $\min_i f_+(i)$, $\min_i f_-(i) \ge \epsilon$ for a constant $\epsilon$.  There exist a Monte Carlo algorithm for \cc~with query complexity $O(\frac{k^2\log{n}}{\Delta(f_+,f_-)})$ with known $f_+$ and $f_-$, and a Las Vegas algorithm with expected query complexity $O(n\log n+\frac{k^2\log{n}}{\Delta(f_+,f_-)})$ even when $k$, $f_+$ and $f_-$ are unknown. (Proof in Sec.~\ref{sec:perfectub}).
 \end{theorem}
 
 To improve from Theorem \ref{thm:mu} to Theorem \ref{thm:div-new}, we would need a more precise approach.  The minor restriction that we have on $f_+$ and $f_-$, namely, $\min_i f_+(i)$, $\min_i f_-(i) \ge \epsilon$ allows $\Delta(f_+,f_-) \leq \frac{2}{\epsilon}$. Note that, by our lower bound result, Lemma \ref{thm:lb-main}, it is not possible to achieve query complexity below $k^2$.
 
 While our lower bound results assume knowledge of $k$, $f_+$ and $f_-$, our Las Vegas algorithms do not even need to know them, and none of the algorithms know $k$. For Theorem \ref{thm:div-new}, indeed, either of $\min_i f_-(i)$ or $\min_i f_+(i)$ having at least $\epsilon$ will serve our purpose. 

The main idea is as follows. It is much easier to determine whether a vertex belongs to a cluster, if that cluster has enough number of members. On the other hand, if a vertex $v$ has the highest membership in some cluster $\calC$ with a suitable definition of membership, then $v$ should be queried with $\calC$ first.  For any vertex $v$ and
 a cluster $\cC$, define the empirical ``inter'' distribution $p_{v,\cC}$ in the following way. For, $i =1,\dots,  q,$:~
 $
 p_{v,\cC}(i) = \frac{1}{|\cC|} \cdot |\{u: w_{u,v} = a_i \}|.
 $ Also compute the `intra' distribution $p_{\cC}$  for $i =1,\dots,  q,$
 $
 p_{\cC}(i) = \frac{1}{|\cC|(|\cC|-1)} \cdot |\{(u,v): u \ne v, w_{u,v} = a_i \}|.
 $
Then {\sf Membership}($v, \cC$) = $- \| p_{v,\cC} -  p_{\cC}\|_{TV}$, where $\| p_{v,\cC} -  p_{\cC}\|_{TV}$ denotes the total variation distance between distributions defined in Section \ref{sec:tool}. If {\sf Membership}($v, \cC$) is highest for $\cC$, then using Sanov's Theorem (Theorem \ref{thm:sanov}) it is highly likely that $v$ is in $\cC$,    if $|\cC|$ is large enough. However we do not know $f_+$ or $f_-$.  Therefore, the highest membership could be misleading since we do not know the desired size threshold that $\cC$ must cross to be reliable. But yet, it is possible to query a few clusters and determine correctly the one which contains $v$. The main reason behind using total variation distance as opposed to divergence, is that divergance is not a metric, and hence do not satisfy the triangle inequality which
becomes crucial in our analysis. This is the precise reason why we need the minimum value to be at least $\epsilon$ in Theorem \ref{thm:div-new}. Under these restrictions, a close relationship between divergence and total variation distance can be established using Pinsker's and Reverse Pinsker's inequalities (see, Section \ref{sec:tool}).


For faulty oracle, let us first take the case of no side information (later, we can combine it with the previous algorithm to obtain similar results with side information). Suppose all $V \times V$ queries have been made. If the  maximum likelihood (ML) estimate on $G$ with these $\binom{n}{2}$ query answers is same as the true clustering of $G$, then Algorithm \ref{algo:cc-error-noside} finds the true clustering with high probability. We sample a small graph $G'$ from $G$, by asking all possible queries in $G'$, and check for the heaviest weight subgraph (assuming $\pm 1$ weight on edges) in $G'$. If that subgraph crosses a desired size, it is removed from $G'$. If this cluster is detected correctly, then it has enough members; we can ask separate queries to them to determine if a vertex belongs to that cluster. The main effort goes in showing that the computed cluster from $G'$ is indeed correct, and that $G'$ has small size. 
\vspace{-0.1in}
\begin{theorem}[Faulty Oracle with No Side Information]
\label{theorem:cc-error}
There exists an algorithm with query complexity $O(\frac{1}{\lambda^2}nk\log{n})$ for \cc~that returns $\hat{G}$, ML estimate of $G$ with all $\binom{n}2$ queries, with high probability when query answers are incorrect with probability $p=\frac{1}{2}-\lambda$. Noting that, $D(p\|1-p) \le \frac{4\lambda^2}{1/2-\lambda}$, this matches the information theoretic lower bound on the query complexity within a $\log{n}$ factor. Moreover, the algorithm returns all the true clusters of $G$ of size at least $\frac{36}{\lambda^2}\log{n}$ with high probability. (Proof in Sec.~\ref{sec:faultyub}).
\end{theorem}

 \begin{theorem}[Faulty Oracle with  Side Information]
 \label{thm:div-new-err}
 Let $f_+,f_-$ be pmfs and $\min_i f_+(i), \min_i f_-(i) \ge \epsilon$ for a constant $\epsilon$.  With side information and faulty oracle with error probability $\frac{1}{2}-\lambda$, there exist an algorithm for \cc~with query complexity $O(\frac{k^2\log{n}}{\lambda^2\Delta(f_+,f_-)})$ when  $f_+,f_-$ known, and an algorithm with expected query complexity $O(n+\frac{k^2\log{n}}{\lambda^2\Delta(f_+,f_-)})$ when  $f_+,f_-$ unknown, that recover $\hat{G}$, ML estimate of $G$ with all $\binom{n}2$ queries, with high probability. {\small (}Proof in Sec.~\ref{sec:faultysideub}{\small )}.
 \end{theorem}

 A subtle part of these results is that, the running time is $O(k^{\frac{\log{n}}{\lambda^2}})$, which is optimal assuming the hardness of planted cliques. However, by increasing the query complexity, the running time can be reduced to polynomial.
 \begin{corollary}[Faulty Oracle with/without Side Information]
 \label{cor:er-poly}
 For faulty oracle with error probability $\frac{1}{2}-\lambda$, there exists a polynomial time algorithm with query complexity $O(\frac{1}{\lambda^2}nk^2)$ for \cc~that recovers all clusters of size at least $O(\max{ \{\frac{1}{\lambda^2}\log{n},k\}})$. (Proof in Sec.~\ref{sec:faultyub}).
 \end{corollary}
 
 As it turns out the ML estimate of $G$ with all $\binom{n}2$ queries is equivalent to computing correlation clustering on $G$ \cite{bm:08,bbc:04,acn:08,cmsy:15,cdk:14}. As a side result, we get a new algorithm for correlation clustering over noisy input, where any cluster of size $\min{(k, \sqrt{n})}$ will be recovered exactly with high probability as long as $k=\Omega(\frac{\log{n}}{\lambda^2})$. When $k \in [\Omega(\frac{\log{n}}{\lambda^2}), o(\sqrt{n})]$, our algorithm strictly improves over \cite{bm:08,bbc:04}.

We hope our work will inspire new algorithmic works in the area of crowdsourcing where both query complexity and side information are important.
\subsubsection{Round Complexity}
Finally, we extend all our algorithms to obtain near optimal round complexity. 

 \begin{theorem}[Perfect Oracle with Side Information]
 \label{thm:perfect-parallel}
  There exists an algorithm for \cc~with perfect oracle and unknown side information $f_+$ and $f_-$ such that it achieves a round complexity within $\tilde{O}(1)$ factor of the optimum when $k=\Omega(\sqrt{n})$ or $k=O(\frac{\sqrt{n}}{\Delta(f_+\|f_-)})$, and otherwise within $\tilde{O}(\frac{1}{{\Delta(f_+\|f_-)}})$.  (Proof in Sec.~\ref{sec:roundpo}).
  \end{theorem}

\begin{theorem}[Faulty Oracle with no Side Information]
\label{thm:error-parallel}
 There exists an algorithm for \cc~with faulty oracle with error probability $\frac{1}{2}-\lambda$ and no side information such that it achieves a round complexity within $\tilde{O}(\sqrt{\log{n}})$ factor of the optimum that recovers $\hat{G}$, ML estimate of $G$ with all $\binom{n}2$ queries with high probability. (Proof in Sec.~\ref{sec:roundfo}).
  \end{theorem}
  This also leads to a new parallel algorithm for correlation clustering over noisy input where computation in every round is bounded by $n\log{n}$.

\section{Organization of the remaining paper}
The rest of the paper is organized as follows. 
In Section \ref{sec:tool}, we provide the information theoretic tools (definitions and basic results) necessary for our 
upper and lower bounds.

In Section \ref{sec:perfect} we provide our main upper bound results for the perfect oracle case when $f_+$ and $f_-$ are unknown.
In Section \ref{sec:res} we give some more insight into the working of Algorithm \ref{algo:cc-exact} and for the case when $f_+$ and $f_-$ are known, provide near optimal Monte Carlo/Las Vegas algorithms
for  \cc~with side information and perfect oracle.
In Section \ref{sec:fault}, we  consider the case when crowd may return erroneous answers. In this scenario we give tight lower and upper bounds on query complexity
in both the cases when we have or lack side information. 
In  Section \ref{sec:round},  we show that the algorithms developed for optimizing query complexity naturally extend to the parallel version of minimizing the round complexity.

 \input{info-theory}

\input{perfect-oracle-v2}

\input{perfect-oracle-residual}

\input{faulty-oracle-v2}

\input{round-complexity}

 \bibliographystyle{plain} 
\bibliography{bibfile}

\end{document}

%% file: info-theory.tex
\section{Information Theory Toolbox}\label{sec:tool}
The lower bounds for randomized algorithms presented in this paper are all information theoretic.  We also use information theoretic tools of large-deviations in upper bounds. To put these bounds into perspective, we will need definition of many information theoretic quantities and some results. Most of this material can also be found in a standard information theory textbook, such as Cover and Thomas \cite{cover2012elements}.
\begin{definition}[Divergence]
The Kullback-Leibler divergence, or simply divergence, between two probability measures  $P$ and $Q$ on a set $\cX$,  is defined to be
$$
D(P\|Q) = \int_{\cX} dP \ln \frac{dP}{dQ}.
$$
When $P$ and $Q$ are distributions of a continuous random variable, represented by probability densities $f_p(x)$ and $f_q(x)$ respectively, we have,
$
D(f_p \| f_q) = \int_{-\infty}^\infty f_p(x) \ln \frac{f_p(x)}{f_q(x)} dx.
$ 
Similarly when $P$ and $Q$ are discrete random variable taking values in the set $\cX$, and represented by the probability mass
functions $p(x)$ and $q(x)$, where $x \in \cX$ respectively, we have
$
D(p(x)\|q(x)) = \sum_{x \in \cX} p(x) \ln \frac{p(x)}{q(x)}.
$
\end{definition}
For two Bernoulli distributions with parameters $p$ and $q$, where $0 \le p,q \le 1$, by abusing the notation 
the divergence is written as,
$$
D(p\|q) = p \ln \frac{p}{q} + (1-p) \ln \frac{1-p}{1-q}.
$$ 
In particular, 
$
D(p\|1-p) = p \ln \frac{p}{1-p} +(1-p) \ln \frac{1-p}{p} = (1-2p) \ln \frac{1-p}{p}.
$
Although $D(P \| Q) \ge 0$, with equality when $P =Q$,  note that in general $D(P\|Q) \ne D(Q\|P)$. Define the {\em symmetric divergence}
between two distribution $P$ and $Q$ as,
$$
\Delta(P, Q) = D(P\|Q) +D(Q\|P).
$$

The following property of the divergence is going to be useful to us.
Consider a set of  random variables $X_1, \dots, X_m$, and consider the two joint distribution of the 
random variables, $P^m$ and $Q^m$. When the random variables are independent, let $P_i$ and $Q_i$ be the corresponding marginal distribution of the 
random variable $X_i, i =1,\dots, m.$  In other words, we have, 
$P^m(x_1, x_2, \dots, x_m ) = \prod_{i =1}^m P_i(x_i)$ and
$Q^m(x_1, x_2, \dots, x_m ) = \prod_{i =1}^m Q_i(x_i).$
Then we must have,
\begin{equation}
D(P^m \| Q^m) = \sum_{i=1}^m D(P_i \| Q_i).
\end{equation}

A more general version, when the random variables are not independent,  is given by the chain-rule, described below for discrete
random variables.
\begin{lemma}\label{lem:chain}
Consider a set of  discrete random variables $X_1, \dots, X_m$, and consider the two joint distribution of the 
random variables, $P$ and $Q$. 
The chain-rule for divergence states that,
$$
D(P(x_1, \dots, x_m)\|Q(x_1, \dots, x_m)) = \sum_{i=1}^m D(P(x_i \mid x_1, \dots, x_{i-1})\|Q(x_i \mid x_1, \dots, x_{i-1})),
$$
where,
$$
D(P(x|y) \| Q(x|y)) = \sum_{y} P(Y=y)D(P(x|Y=y)\|Q(x|Y=y)).
$$
\end{lemma}

\begin{definition}[Total Variation Distance]
For two probability distributions $P$ and $Q$ defined on a sample space $\cX$ and same sigma-algebra $\cF$,
the total variation distance between them is defined to be,
$$
\|P -Q\|_{TV} = \sup \{P(A) -Q(A): A \in \cF\}.
$$
In words, the distance between two distributions is 
their largest difference over any measurable set. For finite $\cX$ total variation distance is half of the $\ell_1$ distance
between pmfs.  
\end{definition}

The total variation distance and the divergence are related by the Pinsker's inequality.
\begin{lemma}[Pinsker's inequality] \label{lem:pinsker}
For any two probability measures $P$ and $Q$,
$$
\|P -Q\|^2_{TV} \le \frac12 D(P\|Q).
$$
\end{lemma}

It is easy to see that, there cannot be a universal `reverse' Pinsker's inequality, i.e., an upper bound on the divergence by
the total variation distance (for example, the total variation distance is always less than 1, while the divergence can be
infinity). However, under various assumptions, such upper bounds have been proposed \cite{sason2015reverse,csiszar2006context}. For example we provide one such
inequality below.

\begin{lemma}[Reverse Pinsker's inequality\cite{sason2015reverse}]\label{lem:revp}
For any two probability measures on finite alphabet $\cX$, given by probability mass functions $p$ and $q$, we must have,
\begin{equation}
\|p-q\|^2_{TV} \ge \frac{\min_{x \in \cX}q(x)}{2}D(p\|q)
\end{equation} 
\end{lemma}
This inequality can be derived from Eq.(28) of \cite{sason2015reverse}.

A particular basic large-deviation inequality that we use for the upper bounds  is Sanov's theorem.
\begin{theorem}[Sanov's theorem] \label{thm:sanov}
Let $X_1, \dots, X_n$ are  iid random variables with a finite sample space $\cX$ and distribution $P$. Let $P^n$ denote their joint distribution.
Let $E$ be a set of probability distributions on $\cX$. The empirical distribution $\tilde{P}_n$ gives probability $\tilde{P}_n(\cA) = \frac{1}{n}\sum_{i=1}^n {\bf 1}_{X_i \in \cA}$ to any event $\cA$. Then,
$$
P^n(\{x_1,\dots, x_n\}: \tilde{P}_n \in E) \le (n+1)^{|\cX|} \exp(-n \min_{P^\ast\in E}D(P^\ast\|P)).
$$
\end{theorem}

A continuous version of Sanov's theorem is also possible but we omit here for clarity.

Hoeffding's  inequality for large deviation of sums  of bounded independent random variables is well known \cite[Thm. 2]{hoeffding1963probability}.
\begin{lemma}[Hoeffding]\label{lem:hoef1}
If $X_1, \dots, X_n$ are  independent random variables   and $a_i\le X_i\le b_i$ for all $i\in [n].$ Then
$$
\Pr(|\frac1n\sum_{i=1}^n (X_i - \avg X_i) | \ge t) \le 2 \exp(-\frac{2n^2t^2}{\sum_{i=1}^n (b_i-a_i)^2}). 
$$
\end{lemma}
This inequality can be used when the random variables are independently sampled with replacement from a finite sample space.  
However due to a result in the same paper  \cite[Thm. 4]{hoeffding1963probability}, this inequality also holds when the random variables are sampled
without replacement from a finite population.
\begin{lemma}[Hoeffding]\label{lem:hoef2}
If $X_1, \dots, X_n$ are  random variables  sampled without replacement from a finite set $\cX \subset \reals$, and $a\le x\le b$ for all $x\in \cX.$ Then
$$
\Pr(|\frac1n\sum_{i=1}^n (X_i - \avg X_i) | \ge t) \le 2 \exp(-\frac{2nt^2}{(b-a)^2}). 
$$
\end{lemma}

%% file: perfect-oracle-v2.tex
\section{\cc~with Perfect Oracle}
\label{sec:perfect}
In this section, we consider the clustering problem using crowdsourcing when crowd always returns the correct answers, and there is side information.
\subsection{Lower Bound}\label{sec:perfectlb}

Recall that  there are $k$ clusters in the $n$-vertex graph. That is $\cG(V,E)$ is such that, $V = \sqcup_{i=1}^k V_i$ and
$E = \{(i,j) : i, j \in V_\ell \text{ for some } \ell\}$. In other words, $\cG$ is a union of at most $k$ disjoint cliques. 
Every entry of the side-information matrix $W$ is generated independently as described in the introduction. We now prove Theorem \ref{thm:lb-main}.

\begin{proof}[Proof of Theorem \ref{thm:lb-main}]
We are going to construct an input that any randomized algorithm will be unable to correctly identify with positive probability. 

Suppose, 
$
a = \Big\lfloor \frac{1}{\Delta(f_+,f_-)}\Big\rfloor.
$
Consider the situation when we are already given a complete cluster $V_k$ with $n - (k-1)a$ elements,  remaining $(k-1)$
clusters each has 1 element, and the rest $(a-1)(k-1)$
 elements are evenly distributed 
(but yet to be assigned) to the $k-1$ clusters. This means each of the smaller clusters has size $a$ each.
Note that, we assumed the knowledge of  the number of clusters $k$.

The side information matrix $W= (w_{i,j})$ is provided. Each $w_{i,j}$ are independent random variables.

Now assume the scenario
when we use an algorithm {\rm ALG} to assigns a vertex to  one of the $k-1$ clusters, $V_u, u =1, \dots, k-1$.
Note that, for any vertex, $l$, the side informations $w_{i,j}$ where $i \ne l$ and $j\ne l$, do not help in assigning $l$ to a cluster (since in that case $w_{i,j}$ is independent of $l$).
Therefore, given a vertex $l$,  {\rm ALG} takes as input the random variables  $w_{i,l}$s where $i \in \sqcup_t V_t$, and
  makes some queries involving $l$ and  outputs a cluster index, which is an assignment for $l$. 
Based on the observations $w_{i,l}$s, the task of algorithm {\rm ALG} is thus a multi-hypothesis testing among $k-1$ hypotheses.
Let $H_u,u =1, \dots k-1$ denote the $k-1$ different hypotheses  $H_u : l \in V_u$. And let  $P_u, u =1, \dots k-1$ denote the joint probability distributions  of the random variables $w_{i,j}$s when $l \in V_u$. In short, for any event $\cA$,
$
P_u(\cA) = \Pr(\cA | H_u) 
$.
Going forward, the subscript of probabilities or expectations will denote the appropriate conditional distribution.

For this hypothesis testing problem, let
$
\avg\{\text{Number of queries made by {\rm ALG}}\} =T.
$
Then, there exist $t \in \{1,\dots, k-1\}$ such that,
$
\avg_t\{\text{Number of queries made by {\rm ALG}}\} \le T.
$
Note that,
\begin{align*}
\sum_{v=1}^{k-1} P_t\{\text{ a query made by {\rm ALG} involving cluster } V_v\} \le \avg_t\{\text{Number of queries made by {\rm ALG}}\} \le T.
\end{align*}

Consider the set
\begin{align*}
J' \equiv \{v\in \{1,\dots,k-1\}:  P_t\{\text{ a query made by {\rm ALG} involving cluster } V_v\} <\frac1{10}\}.
\end{align*}

We must have, 
$
(k-1-|J'|)\cdot \frac1{10} \le T,
$
which implies,
$
|J'| \ge k-1 - 10T.
$




Note that, to output a cluster  without using the side information, {\rm ALG}  has to either make a query to the actual cluster the element is from, or query at least $k-2$ times. In any other case, {\rm ALG} must use the side information (in addition to using queries) to output a cluster.
 Let $\cE^u$ denote the event that {\rm ALG} output cluster  $V_u$ by using the side information. 


Let $J'' \equiv \{u\in \{1,\dots,k-1\}: P_t(\cE^u) \le \frac{10}{k-1} \}.$ 
Since, $\sum_{u=1}^{k-1} P_t(\cE^u) \le 1,$ we must have, 
$$(k-1- |J''|) \cdot \frac{10}{k-1} \le 1,
\text{ or } |J''| \ge \frac{9(k-1)}{10}.$$
$$
\text{We have, }~~~ |J' \cap J''| \ge k-1 - 10T + \frac{9(k-1)}{10} - (k-1) = \frac{9(k-1)}{10} -10T.
$$

Now consider two cases.

\vspace{0.2in}

\noindent{\em Case 1: $T \ge \frac{9(k-1)}{100}$}.
In this case, average number of queries made by {\rm ALG} to assign one vertex to a cluster is at least $ \frac{9(k-1)}{10}$. Since there
are $(k-1)(a-1)$ vertices that needs to be assigned to clusters,  the expected total number of queries performed by {\rm ALG} is 
$\frac{9(k-1)^2(a-1)}{10}$.

\vspace{0.1in}

\noindent{\em Case 2: $T < \frac{9(k-1)}{100}$ }.
In this case, $J' \cap J''$ is nonempty. Assume that we  need to assign  the vertex $j\in V_\ell$ for some $\ell \in J' \cap J''$ to a cluster ($H_\ell$ is the true hypothesis).
We now consider the following two events.
\begin{align*}
\cE_1 &=  \Big\{\text{ a query made by {\rm ALG} involving cluster } V_\ell \Big\}\\
\cE_2 & =  \Big\{ k-2 \text{ or more queries were made by {\rm ALG}} \Big\}.
\end{align*}
Note that, if the algorithm {\rm ALG} can correctly assign $j$ to a cluster without using the side information then either of $\cE_1$
or $\cE_2$ must have to happen. 
Recall, $\cE^\ell$ denote the event that {\rm ALG} output cluster  $V_\ell$  using the side information.
Now consider the event
$
\cE \equiv \cE^\ell \bigcup \cE_1 \bigcup \cE_2.
$
The probability of correct assignment is at most $P_\ell(\cE).$
We  have,
\begin{align*}
P_\ell(\cE) &\le P_t(\cE) + |P_\ell(\cE) - P_t(\cE)|
 \le P_t(\cE) + \|P_\ell - P_t\|_{TV}
 \le  P_t(\cE) + \sqrt{\frac12 D(P_\ell \|P_t)},
 \end{align*}
 where we first used the definition of the total variation distance and  in the last step we have used Pinsker's inequality (Lemma \ref{lem:pinsker}).
Now we bound the divergence $D(P_\ell \|P_t)$. Recall that $P_\ell$ and $P_t$ are the joint distributions of the 
independent random variables $w_{i,j}, i\in \cup_u V_u$. Now, using lemma \ref{lem:chain}, and noting that the 
divergence between identical random variables are $0$, we obtain
$$
D(P_\ell \|P_1) \le a D(f_-\|f_+)+ a D(f_+\|f_-) = a \Delta \le 1.
$$
 This is true because the only times when $w_{i,j}$ differs under $P_t$ and under $P_\ell$ is when $i \in V_t$ or $i \in V_\ell.$
As a result we have,
$
P_\ell(\cE)  \le P_t(\cE) + \sqrt{\frac12}.
$

Now, using Markov inequality $P_t(\cE_2) \le \frac{T}{k-2} \le \frac{9(k-1)}{100(k-2)} \le \frac{9}{100} + \frac{9}{100(k-2)}.$ Therefore, 
\begin{align*}
P_t(\cE)& \le P_t(\cE^\ell) + P_t(\cE_1) +P_t(\cE_2) \le \frac{10}{k-1} +  \frac1{10}+  \frac{9}{100} + \frac{9}{100(k-2)}.
\end{align*}

For large enough $k$, we overall have $
P_\ell(\cE)  \le \frac{19}{100} + \sqrt{\frac12} < \frac{9}{10}.
$
This means ${\rm ALG}$ fails to assign $j$ to the correct cluster 
with probability at least $\frac{1}{10}$. 

\vspace{0.1in}
Considering the above two cases, we can say that any algorithm either makes on average $\frac{9(k-1)^2(a-1)}{10}$ queries, or
makes an error with probability at least $\frac{1}{10}$.
\end{proof}

Note that, in this proof we have not in particular tried to optimize the constants. Corollary \ref{cor:lb-main} follows by noting that to recover the clusters exactly, the query complexity has to be at least $(n-k)+{{k}\choose {2}}$. If the number of queries issued is at most $(n-k)+{{k}\choose {2}}-1$, then either there exists a vertex $v$ in a non-singleton cluster which has not been queried to any other member of that same cluster, or there exist two clusters such that no inter-cluster edge across them have been queried.

\subsection{Upper Bound}\label{sec:perfectub}
 We do not know $k$, $f_+$, $f_-$, $\mu_+$, or $\mu_-$, and our goal, in this section, is to design an algorithm with optimum query complexity for exact reconstruction of the clusters with probability $1$.
We are provided with the side information matrix $W = (w_{i,j})$ as an input. Let $\theta_{gap} = \mu_+ - \mu_-$.

The algorithm uses a subroutine called {\sf  Membership} that takes as input a vertex $v$ and a subset of vertices $\cC\subseteq V.$
At this point, let the {\em membership} of a vertex $v$ in cluster $\cC$ is defined as follows:
$
\average(v, \cC) = \frac{\sum_{u \in \calC} w_{v,u}}{|\calC|},
$
and we use {\sf  Membership}($v, \cC$) = $\average(v, \cC)$.

The pseudocode of the algorithm is given in Algorithm \ref{algo:cc-exact}. The algorithm works as follows. Let $\calC_1, \calC_2,...,\calC_l$ be the current clusters in nonincreasing order of size.  
We find the minimum index $j \in [1,l]$ such that there exists a vertex $v$ not yet clustered, with the highest average membership to $\calC_j$, that is {\sf  Membership}($v, \cC_j$)$ \geq ${\sf  Membership}($v, \cC_{j'}$), $\forall j' \neq j$, and $j$ is the smallest index for which such a $v$ exists. We first check if $v \in \calC_j$ by querying $v$ with any current member of $\calC_j$. If not, then we group the clusters $\calC_1,\calC_2,..,\calC_{j-1}$ in at most $\lceil \log{n} \rceil$ groups such that clusters in group $i$ has size in the range $[\frac{|\calC_1|}{2^{i-1}}, \frac{|\calC_1|}{2^i})$. For each group, we pick the cluster which has the highest average membership with respect to $v$, and check by querying whether $v$ belongs to that cluster. Even after this, if the membership of $v$ is not resolved, then we query $v$ with one member of each of the clusters that we have not checked with previously. If $v$ is still not clustered, then we create a new singleton cluster with $v$ as its sole member.

We now give a proof of the Las Vegas part of Theorem \ref{thm:mu} here using Algorithm \ref{algo:cc-exact}, and defer the more formal discussions on the Monte Carlo part to the next section. 

\begin{proof}[Proof of Theorem \ref{thm:mu}, Las Vegas Algorithm] First, The algorithm never includes a vertex in a cluster without querying it with at least one member of that cluster.  Therefore, the clusters constructed by our algorithm are always proper subsets of the original clusters.  Moreover, the algorithm never creates a new cluster with a vertex $v$ before first querying it with all the existing clusters. Hence, it is not possible that two clusters produced by our algorithm can be merged.

Let $\calC_1,\calC_2,...,\calC_l$ be the current non-empty clusters that are formed by Algorithm \ref{algo:cc-exact}, for some $l \leq k$. Note that Algorithm \ref{algo:cc-exact} does not know $k$. Let without loss of generality $|\calC_1| \geq |\calC_2| \geq ...\geq |\calC_l|$. Let there exists an index $i \leq l$ such that $|\calC_1| \geq |\calC_2| \ge \dots\geq |\calC_{i}| \geq M$,  where $M = \frac{6\log n}{\theta_{gap}^2}$. Of course, the algorithm does not know either $i$ or $M$. If even $|\calC_1| < M$, then $i=0$. Suppose $j'$ is the minimum index such that there exists a vertex $v$ with highest average membership in $\calC_{j'}$. There are few cases to consider based on $j' \leq i$, or $j' > i$ and the cluster  that truly contains $v$.

{\it Case 1. $v$ truly belongs to $\calC_{j'}$.} In that case, we just make one query between $v$ and an existing member of $\calC_{j'}$ and the first query is successful.

{\it Case 2. $j' \leq i$ and $v$ belongs to $\calC_j, j \neq j'$ for some $j \in \{1,\dots, i\}$.} Let $\average(v,\calC_{j})$ and $\average(v,\calC_{j'})$ be the average membership of $v$ to $\calC_j$, and $\calC_{j'}$ respectively. Then we have $\average(v,\calC_{j'}) \geq \average(v,\calC_{j})$, that is {\sf  Membership}($v, \cC_j'$)$ \geq ${\sf  Membership}($v, \cC_{j}$). This is only possible if either $\average(v,\calC_{j'}) \geq \mu_r+\frac{\theta_{gap}}{2}$ or $\average(v,\calC_{j}) \leq \mu_g-\frac{\theta_{gap}}{2}$. Since both $\calC_j$ and $\calC_j'$ have at least $M$ current members, then using the Chernoff-Hoeffding's bound (Lemma \ref{lem:hoef1}) followed by union bound this happens with probability at most $\frac{2}{n^3}$. Therefore, the expected number of queries involving $v$ before its membership gets determined is $\leq 1+\frac{2}{n^3} k < 2$.

{\it Case 4. $v$ belongs to $\calC_j, j \neq j'$ for some $j > i$.} In this case the algorithm may make $k$ queries involving $v$ before its membership gets determined.

{\it Case 5. $j' > i$, and $v$ belongs to $\calC_j$ for some $j \leq i$.} 
In this case, there exists no $v$ with its highest membership in $\calC_1,\calC_2,...,\calC_i$.

Suppose $\calC_1,\calC_2,...,\calC_j'$ are contained in groups $H_1,H_2,...,H_s$ where $s \leq \lceil \log{n}\rceil$. Let $\calC_j \in H_t$, $t \in [1,s]$. Therefore, $|\calC_j| \in [\frac{|\calC_1|}{2^{t-1}}, \frac{|\calC_1|}{2^{t}}]$. If $|\calC_j| \geq 2M$, then all the clusters in group $H_t$ have size at least $M$. Now with probability at least $1-\frac{2}{n^2}$, $\average(v,\calC_j) \geq \average(v,\calC_{j''})$, that is {\sf  Membership}($v, \cC_j$)$ \geq ${\sf  Membership}($v, \cC_{j''}$) for every cluster $\calC_{j''} \in H_t$. In that case, the membership of $v$ is determined within at most $\lceil \log{n}\rceil$ queries. Otherwise, with probability at most $\frac{2}{n^2}$, there may be $k$ queries to determine the membership of $v$.

Therefore, once a cluster has grown to size $2M$, the number of queries to resolve the membership of any vertex in those clusters is at most $\lceil \log{n} \rceil$ with probability at least $1-\frac{2}{n}$. Hence, for at most $2kM$ elements, the number of queries made to resolve their membership can be $k$. Thus the expected number of queries made by Algorithm \ref{algo:cc-exact} is 
$O(n\log{n}+Mk^2)=O(n\log{n}+\frac{k^2\log{n}}{(\mu_+-\mu_{-})^2})$. Moreover, if we knew $\mu_{+}$ and $\mu_{-}$, we can calculate $M$, and thus whenever a clusters grows to size $M$, remaining of its members can be included in that cluster without making any error with high probability. This leads to Theorem \ref{thm:mu}.
\end{proof}

We can strengthen this algorithm by changing the subroutine {\sf Membership} in the following way. Assume that $f_+, f_-$ are discrete distributions over $q$ points $a_1, a_2, \dots, a_q$; that is $w_{i,j}$ takes value in the set $\{a_1, a_2, \dots, a_q\}\subset [0,1].$


The subroutine {\sf Membership} takes $v \in V$ and $\cC \subseteq V\setminus\{v\}$ as inputs.
Compute the `inter' distribution  $p_{v,\cC}$ for $i =1,\dots,  q,$
 $
 p_{v,\cC}(i) = \frac{1}{|\cC|} \cdot |\{u: w_{u,v} = a_i \}|.
 $
 
 Also compute the `intra' distribution $p_{\cC}$  for $i =1,\dots,  q,$
 $
 p_{\cC}(i) = \frac{1}{|\cC|(|\cC|-1)} \cdot |\{(u,v): u \ne v, w_{u,v} = a_i \}|.
 $
Then define {\sf Membership}($v, \cC$) = $- \| p_{v,\cC} -  p_{\cC}\|_{TV}.$ Note that, since the membership is 
always negative, a higher membership implies that the `inter' and `intra' distributions are closer in terms
of total variation distance. 
With this modification in the subroutine we can prove what is claimed in Theorem \ref{thm:div-new}.

The analysis for this case proceeds exactly as above. However, to compare memberships we use Lemma \ref{lem:unknown} below.
Indeed, Lemma \ref{lem:unknown} can be used in the cases 2 and 5 in lieu of Chernoff-Hoeffding bounds to obtain the exact same result.
\begin{lemma}\label{lem:unknown}
Suppose, $\cC, \cC' \subseteq V$, $\cC \cap \cC' = \emptyset$ and $|\cC| \ge M, |\cC'| \ge M =\frac{16 \log n}{\epsilon \Delta(f_+,f_-)}$, with where $\min_i f_+(i), \min_i f_-(i) \ge \epsilon$ for a constant $\epsilon$. Then, 
$$
\Pr\Big({\sf Membership}(v, \cC') \ge {\sf Membership}(v, \cC) \mid v \in \cC\Big) \le \frac{4}{n^3}.
$$
\end{lemma}
\begin{proof}
Let $\beta = \frac{\|f_+-f_-\|_{TV}}{2}$. If ${\sf Membership}(v, \cC') \ge {\sf Membership}(v, \cC)$ then we must have,
$\|p_{v,\cC'} - p_{\cC'}\|_{TV} \le \|p_{v,\cC} - p_{\cC}\|_{TV} .$ This means, either
$
\|p_{v,\cC'} - p_{\cC'}\|_{TV}  \le \frac{\beta}{2}
$
or $ \|p_{v,\cC} - p_{\cC}\|_{TV}  \ge \frac{\beta}{2}.$ Now, using triangle inequality, 
\begin{align*}
&\Pr\Big(\|p_{v,\cC'} - p_{\cC'}\|_{TV}  \le \frac{\beta}{2} \Big) \le \Pr\Big(\|p_{v,\cC'} - f_+\|_{TV}  - \|p_{\cC'} - f_+\|_{TV}  \le \frac{\beta}2\Big)\\
&  \le \Pr\Big(\|p_{v,\cC'} - f_+\|_{TV}  \le \beta \text{ or } \|p_{\cC'} - f_+\|_{TV}  \ge \frac{\beta}2\Big)
\le  \Pr\Big(\|p_{v,\cC'} - f_+\|_{TV}  \le \beta\Big) + \Pr\Big(\|p_{\cC'} - f_+\|_{TV}  \ge \frac{\beta}2\Big).
\end{align*}
Similarly,
\begin{align*}
&\Pr\Big(\|p_{v,\cC} - p_{\cC}\|_{TV}  \ge \frac{\beta}{2} \Big) \le \Pr\Big(\|p_{v,\cC} - f_+\|_{TV}  + \|p_{\cC} - f_+\|_{TV}  \ge \frac{\beta}2\Big)\\
&  \le \Pr\Big(\|p_{v,\cC} - f_+\|_{TV}   \ge \frac{\beta}4 \text{ or }  \|p_{\cC} - f_+\|_{TV}  \ge \frac{\beta}4\Big)
\le  \Pr\Big(\|p_{v,\cC} - f_+\|_{TV}   \ge \frac{\beta}4\Big) + \Pr\Big( \|p_{\cC} - f_+\|_{TV}  \ge \frac{\beta}4\Big).
\end{align*}
Now, using Sanov's theorem (Theorem \ref{thm:sanov}), we have,
$$
 \Pr\Big(\|p_{v,\cC'} - f_+\|_{TV}  \le \beta\Big) \le (M+1)^q \exp(-M \underset{p: \|p - f_+\|_{TV}  \le \beta}\min D(p\|f_-)).
$$
At the optimizing $p$ of the exponent, 
\begin{align*}
D(p\|f_-)& \ge 2\|p -f_-\|_{TV}^2 & \text{ from Pinsker's Inequality (Lemma \ref{lem:pinsker})} \\
& \ge 2(\|f_+-f_-\|_{TV} - \|p-f_+\|_{TV})^2 & \text{ from using triangle inequality} \\
& \ge 2(2\beta - \beta)^2 & \text{ from noting the value of $\beta$}\\ 
&= \frac{\|f_+-f_-\|_{TV}^2}{2}\\
& \ge \frac{\epsilon}2 \max\{D(f_+\|f_-), D(f_-\|f_+) \}  &\text{from reverse Pinsker's inequality (Lemma \ref{lem:revp})}\\
&\ge \frac{\epsilon \Delta(f_+,f_-)}{4}
\end{align*}

Again, using Sanov's theorem (Theorem \ref{thm:sanov}), we have,
$$
\Pr\Big(\|p_{\cC'} - f_+\|_{TV}  \ge \frac{\beta}2\Big) \le (M+1)^q \exp(-M \underset{p: \|p - f_+\|_{TV}  \ge \frac{\beta}{2}}\min D(p\|f_+)).
$$
At the optimizing $p$ of the exponent, 
\begin{align*}
D(p\|f_+)& \ge 2\|p -f_+\|_{TV}^2 & \text{ from Pinsker's Inequality (Lemma \ref{lem:pinsker})} \\
& \ge \frac{\beta^2}{2} & \text{ from noting the value of $\beta$}\\ 
&= \frac{\|f_+-f_-\|_{TV}^2}{8}\\
& \ge \frac{\epsilon}8 \max\{D(f_+\|f_-), D(f_-\|f_+) \}  &\text{from reverse Pinsker's inequality (Lemma \ref{lem:revp})}\\
&\ge \frac{\epsilon \Delta(f_+,f_-)}{16}
\end{align*}

Now substituting this in the exponent,
 using the value of $M$ and doing the same exercise for the other two probabilities we get the claim of the lemma.
\end{proof}

\algnewcommand{\LineComment}[1]{\Statex \(\triangleright\) \emph{#1}}
\begin{algorithm}
\caption{\cc~with Side Information. Input: $\{V,W\}$ (Note: $\cO$ is the perfect oracle.}
\label{algo:cc-exact}
\begin{algorithmic}[1]
\LineComment{Initialization.}
\State Pick an arbitrary vertex $v$ and create a new cluster $\{v\}$. Set $V=V\setminus v$
\While{$V\neq \emptyset$}
\LineComment{Let the number of current clusters be $l \geq 1$}
\State Order the existing clusters in nonincreasing size. 
\LineComment{Let $|\calC_{1}| \geq |\calC_{2}| \geq \ldots \geq |\calC_{l}|$ be the ordering (w.l.o.g).}
\For{$j=1$ to $l$ }
\State If $\exists v \in V$ such that $j=\max_{i\in[1,l]} {\sf Membership}(v, \cC_i)$, then select $v$ and Break;
\EndFor
\State $\cO(v,u)$ where $u \in \calC_{j}$
\If{$\cO(v,u)==``+1"$}
\State Include $v$ in $\calC_{j}$. $V=V \setminus v$
\Else
\LineComment{logarithmic search for membership in the large groups. Note $s \leq \lceil \log{n} \rceil$}
\State Group $\calC_1,\calC_2,...,\calC_{j-1}$ into $s$ consecutive classes $H_1, H_2,...,H_s$ such that the clusters in group $H_i$ have their current sizes in the range $[\frac{|\calC_1|}{2^{i-1}}, \frac{|\calC_1|}{2^i})$
\For{$i=1$ to $s$}
\State $j=\max_{a:\calC_a \in H_i}  {\sf Membership}(v, \cC_a)$
\State $\cO(v,u)$ where $u \in \calC_j$. 
\If{$\cO(v,u)==``+1"$}
\State Include $v$ in $\calC_{j}$. $V=V \setminus v$. Break.
\EndIf
\EndFor
\LineComment{exhaustive search for membership in the remaining groups}
\If{ $v \in V$}
\For{$i=1$ to $l+1$}
	\If{$i=l+1$} 
	\Comment{$v$ {\em does not belong to any of the existing clusters}}
	\State Create a new cluster $\{v\}$. Set $V=V\setminus v$
\Else
\If{ $\nexists u \in \calC_i$ such that $(u,v)$ has already been queried}
\State $\cO(v,u)$
\If{$\cO(v,u)==``+1"$}
\State Include $v$ in $\calC_{j}$. $V=V \setminus v$. Break.
\EndIf
\EndIf
\EndIf
\EndFor
\EndIf
\EndIf
\EndWhile
\end{algorithmic}
\end{algorithm}

%% file: perfect-oracle-residual.tex
\section{\cc~with Perfect Oracle: Known $f_+, f_-$ }\label{sec:res}
In this section, we take a closer look at the results presented in Section \ref{sec:perfect}. 
 Recall that, we have an undirected graph $G(V\equiv[n],E)$, such that $G$ is a union of $k$ disjoint cliques $G_i(V_i, E_i)$, $i =1, \dots, k$, but the subsets $V_i \in [n]$ and $E$ are unknown to us. The goal is to determine these clusters accurately (with probability $1$) by making minimum number of pair-wise queries.
 As a side information, we are given $W$ which represents the similarity values that are computed by some automated algorithm, and therefore reflects only a noisy version of the true similarities ($\{0, 1\}$). Based on the sophistication of the automated algorithm, and the amount of information available to it, the densities $f_{+}$ and $f_{-}$ will vary. We have provided a lower bound for this in Section \ref{sec:perfect}.

In Algorithm \ref{algo:cc-exact}, we do not know $k$, $f_{+}$, $f_{-}$, $\mu_+$, or $\mu_-$, and our goal was to achieve optimum query complexity for exact reconstruction of the clusters with probability $1$.
In this section we are going to provide a simpler algorithm that has the knowledge of $\mu_+$, $\mu_-$ or even 
$f_{+}$, $f_{-}$, and show that we can achieve optimal query complexity.

 Let $\mu_+-\mu_- \geq \theta_{gap}$, and  we select a parameter $M$ satisfying 
$$M = \frac{6\log n}{\theta_{gap}^2}.$$
The simpler algorithm, referred to as Algorithm (\ref{algo:cc-exact}-a), contains  two phases, that are repeated as long as there are vertices that have not yet been clustered.

{\bf Querying Phase.} The algorithm maintains a list of active clusters which contain at least one vertex but whose current size is strictly less than $M$. For every $v$ which has not yet been assigned to any cluster, the algorithm checks by querying to the oracle whether $v$ belongs to any of the cluster in the list. If not, it opens a new cluster with $v$ as its sole member, and add that cluster to the list.

{\bf Estimation Phase.} If the size of any cluster, say $\calC$ in the list becomes $M$, then the cluster is removed from the list, and the algorithm enters an estimation phase with $\calC$. For every vertex $v$ which has not yet been clustered, it computes the average membership score of $v$ in $\calC$ as $\average(v,\calC)=\frac{1}{|\calC|}\sum_{u} w_{u,v}$. If $\average(v,\calC) \geq \mu_+-\frac{\theta_{gap}}{2}$, then include $v$ in $\calC$. After this phase, mark $\calC$ as final and inactive. 

\begin{lemma}
The total number of queries made by Algorithm (\ref{algo:cc-exact}-a) is at most $k^2M$.
\end{lemma}
\begin{proof}
Suppose, there are $k'$ clusters of size at least $M$. The number of queries made in the querying phase to populate these clusters is at most $Mk'\cdot k$. For the remaining $(k-k')$ clusters, their size is at most $(M-1)$, and again the number of queries made to populate them is at most $(M-1)(k-k')\cdot k$. Hence, the total number of queries made during the querying phases is $k^2M$. Furthermore, no queries are made during the estimation phases, and we get the desired bound.
\end{proof}

\begin{lemma}
Algorithm (\ref{algo:cc-exact}-a) retrieves the original clusters with probability at least $1-\frac{2}{n}$.
\end{lemma}
\begin{proof}
Any vertex that is included to an active cluster, must have got a positive answer from a query during the querying phase. Now consider the estimation phase. If $u, v \in V$ belong to the same cluster $\calC$, then $w_{u,v} \sim f_{+}$, else $w_{u,v} \sim f_{-}$. Therefore, $\avg{[w_{u,v} \mid u,v \in \calC]}=\mu_+$ and $\avg{[w_{u,v} \mid u \in \calC ,v \notin \calC]}=\mu_-$. Then, by the Chernoff-Hoeffding bound (Lemma \ref{lem:hoef1}), 
$$\Pr\Big(\average(v,\calC) < \mu_+-\frac{\theta_{gap}}{2}| v \in \calC\Big) \le  
e^{-\frac{M \theta_{gap}^2}{2}} \le \frac{1}{n^3}.$$

And similarly, 

$$\Pr\Big(\average(v,\calC) > \mu_-+\frac{\theta_{gap}}{2} \mid v \not\in \calC \Big)\le 
e^{-\frac{M \theta_{gap}^2}{2}} \le \frac{1}{n^3}.$$

Therefore by union bound, for every vertex that is included in $\calC$ during the estimation phase truly belongs to it, and any vertex that is not included in $\calC$ truly does not belong to it with probability $\geq 1-\frac{2}{n^2}$. Or, the probability that the cluster $\calC$ is not correctly constructed is at most $\frac{2}{n^2}$. Since, there could be at most $n$ clusters, the probability that there exists one incorrectly constructed cluster is at most $\frac{2}{n}$. Note that, any cluster that never enters the estimation phase is always constructed correctly, and any cluster that enters the estimation phase is fully constructed before moving to a new unclustered vertex. Therefore, if a new cluster is formed in the querying phase with $v$, then $v$ cannot be included to any existing clusters assuming the clusters grown in the estimation phases are correct. Hence, Algorithm  (\ref{algo:cc-exact}-a) correctly retrieves the clusters with probability at least $1-\frac{2}{n}$.
\end{proof}

So far, Algorithm (\ref{algo:cc-exact}-a) has been a Monte Carlo algorithm. In order to turn it into a Las Vegas algorithm, we make the following modifications.
\begin{itemize}
\item In the {\sc estimation phase} with cluster $\calC$, if $\average{(v,C)} \geq \mu_+-\frac{\theta_{gap}}{2}$, then we query $v$ with some member of $\calC$. If that returns $+1$ (i.e., the edge is present), then we include $v$, else, we query $v$ with one member of every remaining clusters (active and inactive). If none of these queries returns $+1$, then a singleton cluster with $v$ is created, and included in the active list. We then proceed to the next vertex in the estimation phase.
\end{itemize}

Clearly, this modified Algorithm (\ref{algo:cc-exact}-a) retrieves all clusters correctly and it is a Las Vegas algorithm. We now analyze the expected number of queries made by the algorithm in the estimation phase.

\begin{lemma}
The modified Las Vegas Algorithm (\ref{algo:cc-exact}-a) makes at most $n+2$ queries on expectation during the estimation phase.
\end{lemma}
\begin{proof}
In the estimation phase, only $1$ query is made with $v$ while determining its membership if indeed $v$ belongs to cluster $\calC$ when $\average{(v,C)} \geq \mu_+-\frac{\theta_{gap}}{2}$. Now this happens with probability at least $1-\frac{2}{n^2}$. With the remaining probability, at most $(k-1) \leq n$ extra queries may be made with $v$. At the end of this, either $v$ is included in a cluster, or a new singleton cluster is formed with $v$. Therefore, the expected number of queries made with $v$, at the end of which the membership of $v$ is determined is at most $1+\frac{2(k-1)}{n^2}$. Hence the expected number of total queries made by Algorithm (\ref{algo:cc-exact}-a) in the estimation phase is at most $n+2$. \end{proof}

\begin{theorem}
With known $\mu_+$ and $\mu_-$, there exist a Monte Carlo algorithm for \cc~with query complexity $O(\frac{k^2\log{n}}{(\mu_+-\mu_-)^2})$ and a Las Vegas algorithm with expected query complexity $O(n+\frac{k^2\log{n}}{(\mu_+-\mu_-)^2})$. 
\end{theorem}

\paragraph*{Comparing with the Lower Bound.} 
\begin{example}
The KL-divergence between two univariate normal distributions with means $\mu_1$ and $\mu_2$, and standard deviations $\sigma_1$ and $\sigma_2$ respectively can be calculated as $D(\mathcal{N}(\mu_1,\sigma_1)\|\mathcal{N}(\mu_2,\sigma_2))=\log{\left(\frac{\sigma_1}{\sigma_2}\right)}+\frac{\sigma_1^2+(\mu_1-\mu_2)^2}{2\sigma_2^2}-\frac{1}{2}$. Therefore $\Delta(f_{+},f_{-})=(\mu_+-\mu_-)^2$. Algorithm (\ref{algo:cc-exact}-a) is optimal under these natural distributions within a $\log{n}$ factor.
\end{example}
\begin{example}
\[ f_{-}(x) =
  \begin{cases}
    (1+\epsilon)       & \quad \text{if } 0\leq x < \frac{1}{2}\\
    (1-\epsilon)  & \quad \text{if } 1 \geq x \geq \frac{1}{2};
  \end{cases}
\]

\[ f_{+}(x) =
  \begin{cases}
    (1-\epsilon)       & \quad \text{if } 0 \leq x < \frac{1}{2}\\
    (1+\epsilon)  & \quad \text{if } 1 \geq x \geq \frac{1}{2}.
  \end{cases}
\]

That is, they are derived by perturbing the uniform distribution slightly so that $f_{+}$ puts slightly higher mass when $x \geq \frac{1}{2}$, and $f_{-}$ puts slightly higher mass when $x < \frac{1}{2}$.

Note that $\int_{0}^{1} f_{-}(x) \, \diff x=\int_{0}^{1/2} (1+\epsilon)\,  \diff x+\int_{1/2}^{1} (1-\epsilon)\,  \diff x=1$. Similarly, $\int_{0}^{1} f_{+}(x) \, \diff x=1$, that is they represent valid probability density functions.

We have $$\mu_-=\int_{0}^{1} x f_{-}(x) \, \diff x=\frac{(1+\epsilon)}{8}+\frac{3(1-\epsilon)}{8}=\frac{2-\epsilon}{4}=\frac{1}{2}-\frac{\epsilon}{4}$$
and
$$\mu_+=\int_{0}^{1} x f_{+}(x) \, \diff x=\frac{(1-\epsilon)}{8}+\frac{3(1+\epsilon)}{8}=\frac{2+\epsilon}{4}=\frac{1}{2}+\frac{\epsilon}{4}.$$

Thereby, $\mu_+-\mu_-=\frac{\epsilon}{2}$.
Moreover \begin{align*}
D(f_{+}\|f_{-})=\int_{0}^{1/2} (1-\epsilon) \log{\frac{1-\epsilon}{1+\epsilon}}dx+\int_{1/2}^{1} (1+\epsilon) \log{\frac{1+\epsilon}{1-\epsilon}}dx=\epsilon \log{\frac{1+\epsilon}{1-\epsilon}}=O(\epsilon^2)
\end{align*}

Therefore, again Algorithm (\ref{algo:cc-exact}-a) is optimal under these distributions within a $\log{n}$ factor.

\end{example}

\paragraph*{Improving Algorithm (\ref{algo:cc-exact}-a) to match the lower bound.}

 We will now show a way to achieve the lower bound in the \cc~  up to a logarithmic term (while matching the denominator)
 by modifying Algorithm (\ref{algo:cc-exact}-a). We first do this by assuming $f_{+}, f_{-}$ to be discrete distributions
 over $q$ points $a_1, a_2, \dots, a_q$. 
 So, $w_{i,j}$ takes value in the set $\{a_1, a_2, \dots, a_q\}$.
 
 \begin{theorem}\label{thm:div}
 With known $f_+$ and $f_-$ such that $\min_i f_+(i), \min_i f_-(i) \ge \epsilon$ for a constant $\epsilon$,  there exist a Monte Carlo algorithm for \cc~with query complexity $O(\frac{k^2\log{n}}{\Delta(f_+,f_-)})$ and a Las Vegas algorithm with expected query complexity $O(n+\frac{k^2\log{n}}{\Delta(f_+,f_-)})$. 
 \end{theorem}

 Indeed, either of $\min_i f_-(i)$ or $\min_i f_+(i)$ strictly greater than $0$ will serve our purpose. We have argued before that  it is not
so restrictive condition. 

 \begin{proof}[Proof of Theorem \ref{thm:div}]
  For any vertex $v$ and
 a cluster $\cC$, define the empirical distribution $p_{v,\cC}$ in the following way. For, $i =1,\dots,  q,$
 $$
 p_{v,\cC}(i) = \frac{1}{|\cC|} \cdot |\{u: w_{u,v} = a_i \}|.
 $$
Now modify Algorithm (\ref{algo:cc-exact}-a) as follows. The querying phase of the algorithm remains exactly same. In the
estimation phase for a cluster $\cC$ and an unassigned vertex $v$, include $v$ in $\cC$ if 
$$
D(p_{v,\cC} \| f_+) < D(p_{v,\cC} \| f_-).
$$
Everything else remains same.

Now, a vertex $v \in \cC$ will be erroneously not assigned to it with probability
\begin{align*}
\Pr\Big(D(p_{v,\cC} \| f_+) \ge D(p_{v,\cC} \| f_-) \mid v \in \cC\Big)& = f_+\Big(\{p_{v,\cC}: D(p_{v,\cC} \| f_+) \ge D(p_{v,\cC} \| f_-) \}\Big)\\
& \hspace{-2in} = (M+1)^{q} \exp(-M \underset{p: D(p\| f_+) \ge D(p\| f_-)}\min D(p\| f_+)),
\end{align*}
where in the last step we have used Sanov's theorem (see, Theorem \ref{thm:sanov}). Due to lemma \ref{lem:cvxopt}, we can replace the constraint of the optimization in the exponent above by an equality.

Hence,
$$
\Pr\Big(D(p_{v,\cC} \| f_+) \ge D(p_{v,\cC} \| f_-) \mid v \in \cC\Big) \le (M+1)^{q+1} \exp({-M \underset{{p: D(p\| f_+) = D(p\| f_-)}}\min D(p\| f_+)}) \le \frac{1}{n^3},
$$
whenever $$M = \frac{8\log n}{\underset{p: D(p\| f_+) = D(p\| f_-)} \min D(p\| f_+)}.$$
This value of $M$ is also sufficient to have,
$$
\Pr\Big(D(p_{v,\cC} \| f_+) < D(p_{v,\cC} \| f_-) \mid v \notin \cC\Big) \le \frac{1}{n^3}.
$$
While the rest of the analysis stays same as before, the overall query complexity of this modified algorithm is
$$
O\Big(\frac{k^2 \log n}{\underset{p: D(p\| f_+) = D(p\| f_-)}\min  D(p\| f_+)}\Big).
$$
If the divergence were a distance then in the denominator above we would have got $D(f_+\|f_-)/2$ and that
would be same as the lower bound we have obtained. However, since that is not the case, we rely on the following chain of
inequalities instead at the optimizing point of $p$.
\begin{align*}
&D(p\|f_+) = D(p\|f_-) = \frac{D(p\|f_+)+D(p\|f_-)}2 \ge \|p-f_+\|_{TV}^2 + \|p-f_-\|^2_{TV}\\
& \qquad \ge \frac{( \|p-f_+\|_{TV} +\|p-f_-\|_{TV})^2 }{2} \ge \frac{\|f_+-f_-\|_{TV}^2}{2},
\end{align*}
where we have  used the Pinsker's inequality (Lemma \ref{lem:pinsker}), the convexity of the function $x^2$ and the 
triangle inequality for the total variation distance respectively. Now as the last step we use the reverse Pinsker's inequality (Lemma \ref{lem:revp})
to obtain,
$$
\|f_+-f_-\|_{TV}^2 \ge \frac{\epsilon}{2} \max\{D(f_+\|f_-), D(f_-\|f_+) \} \ge \frac{\epsilon\Delta(f_+,f_-)}{2}.
$$
This completes the proof.
\end{proof}

\begin{lemma}\label{lem:cvxopt}
$$
\underset{p: D(p\| f_+) \ge D(p\| f_-)}\min D(p\| f_+) = \underset{p: D(p\| f_+) = D(p\| f_-)}\min D(p\| f_+).
$$
\end{lemma}

\begin{proof}
Since the condition $D(p\| f_+) \ge D(p\| f_-)$ can be written as $\sum_{i} p(i) \ln \frac{f_-(i)}{f_+(i)} \ge 0$ we need to solve the constrained optimization
\begin{equation}\label{eq:divm}
\min_p D(p\| f_+)
\end{equation}
such that
\begin{equation}\label{eq:const}
\sum_{i} p(i) \ln \frac{f_-(i)}{f_+(i)} \ge 0.
\end{equation}
We claim that the inequality of \eqref{eq:const} can be replaced by an equality without any change in the optimizing value. Suppose this is not true and the optimizing value $\tilde{p}$ is such that 
$\sum_{i} \tilde{p}(i) \ln \frac{f_-(i)}{f_+(i)} = \epsilon >0$. 

Let, $\lambda = \frac{\epsilon}{\epsilon + D(f_-\|f_+)} \in (0,1)$.  Note that for the value $p^\ast = \lambda f_+ + (1-\lambda)\tilde{p}$ we have,
$$s
\sum_{i} p^\ast(i) \ln \frac{f_-(i)}{f_+(i)} = - \lambda D(f_+\|f_-) + (1-\lambda) \epsilon =0.
$$
However, since $D(p\| f_+)$ is a strictly convex function of $p$, we must have,
$$
D(p^\ast\| f_+) <  \lambda D(f_+\| f_+) + (1-\lambda) D(\tilde{p}\|f_+) =  (1-\lambda) D(\tilde{p}\|f_+),
$$
which is a contradiction of $\tilde{p}$ being the optimizing value.

\end{proof}

%% file: faulty-oracle-v2.tex
\newtheorem{conjecture}{Conjecture}
\section{\cc~with Faulty Oracle}
\label{sec:fault}
We now consider the case when crowd may return erroneous answers. We do not allow resampling of the same query. By resampling one can always  get correct answer for each query with high probability followed by which we can simply apply the algorithms for the perfect oracle. Hence, the oracle can be queried with a particular tuple only once in our setting.
%

\subsection{No Side Information}
\subsubsection{Lower bound for the faulty-oracle model}
\label{sec:error-lc}

Suppose, $G(V, E)$ is a union of $k$ disjoint cliques as before. We have,  $V = \sqcup_{i=1}^k V_i$. 
We consider the following faulty-oracle model. We can query the oracle whether there exists an edge between vertex $i$ and $j$. The oracle will give the correct answer with probability $1-p$ and will give the incorrect answer with probability $p$. We would like to estimate the minimum number of queries one must make to the oracle so that we can recover the clusters with high probability. In this section we forbid the use of any side information 
that may be obtained from an automated system. The main goal of this section is to prove Theorem \ref{thm:faulty}.

As argued in the introduction, there is a need for a  minimum cluster size. If there is no minimum size requirement on a cluster then the input graph can always consist of multiple clusters of very small size. Then consider the following two different clusterings.
$C_1: V = \sqcup_{i=1}^{k-2} V_i \sqcup\{v_1,v_2\}\sqcup\{v_3\}$ and $C_2: V =   \sqcup_{i=1}^{k-2} V_i \sqcup\{v_1\}\sqcup\{v_2,v_3\}$. Now if one of these two clusterings are given to us uniformly at random, no matter how many queries we do, we will fail to recover the correct cluster with probability at least $p$ (recall that, resampling is not allowed).

The argument above does not hold for the case when $p=0$. In that case any (randomized) algorithm has to  use  (on expectation) $O(nk)$ queries for correct clustering (see the $p=0$ case of Theorem \ref{thm:faulty} below). While for deterministic algorithm the proof of the above fact is straight-forward, for randomized algorithms it  was established in \cite{dkmr:14}. In \cite{dkmr:14}, a clustering was called {\em balanced} if the minimum and maximum sizes of the clusters are only a constant factor way. In particular, \cite{dkmr:14} observes that, for unbalanced 
input the lower bound for $p=0$ case is easier. For randomized algorithms  and balanced inputs, they left the lower bound as an open problem. Theorem \ref{thm:faulty} resolves this as a special case.

Indeed, in Theorem  \ref{thm:faulty}, we provide lower bounds for $0\le p <1/2$,  assuming inputs such that either 1) the maximum size of the cluster is within a constant times away from the average size, or 2) the minimum size of the cluster is  a constant fraction of the average size. Note that the average size of a cluster  is $\frac{n}{k}$.

The technique to prove Theorem  \ref{thm:faulty} is similar to  the one we have used in Theorem \ref{thm:lb-main}. However  
we only handle binary random variables here (the answer to the queries).
The significant difference is that, while designing the  input we consider a {\em balanced} clustering 
with small sized clusters we can always fool any algorithm as 
exemplified above). This removes  the constraint on the algorithm designer on how many times a cluster can be queried with a vertex.
While Lemma \ref{lem:clus} shows that enough number of queries must be made with a large number of vertices $V' \subset V$, either of the conditions on minimum or maximum sizes of a cluster ensures that $V'$ contains enough vertices that do not satisfy this query requirement.


  
As mentioned,  Lemma \ref{lem:clus} is crucial  to prove Theorem  \ref{thm:faulty}.

\begin{lemma}\label{lem:clus}
Suppose, there are $k$ clusters. There exist at least $\frac{4k}{5}$ clusters such that  a vertex $v$ from any one of these clusters will be assigned to a wrong cluster by any randomized algorithm with positive probability unless the  number  of queries involving $v$ is more than $\frac{k}{10D(p\|1-p)}$ when $p >0$ and $\frac{k}{10}$ when $p =0$. 
\end{lemma}
\begin{proof}
Let us assume that the $k$ clusters are already formed, and we can moreover assume that all vertices except for the said vertex  has already been assigned to a cluster. 
Note that, queries that do not involve the said vertex plays no role in this stage.

Now the problem reduces to a hypothesis testing problem where the $i$th hypothesis 
$H_i$ 
for $i =1,\dots, k$, denote that the true cluster is $V_i$. We can also add a null-hypothesis $H_0$ that stands for the vertex belonging to none of the clusters. Let $P_i$ denote the joint probability distribution of our observations (the answers to the queries involving vertex $v$) when $H_i$ is true, $i =0,1,\dots, k$. That is for any event $\cA$ we have,
$$
P_i(\cA) = \Pr(\cA | H_i).
$$

Suppose $Q$   denotes the total  number of queries made by a (possibly randomized) algorithm at this stage. 
Let the random variable $Q_i$ denote  the number of queries involving cluster $V_i, i =1,\dots, k.$

We must have,
$
\sum_{i=1}^k \avg_0 Q_i \le Q.
$
Let, $$J_1 \equiv \{i \in \{1,\dots, k\}: \avg_0 Q_i \le \frac{10Q}k  \}.$$ Since,
$
(k -|J_1|) \frac{10Q}k \le Q,
$
we have $|J_1| \ge \frac{9k}{10}$.

Let $\cE_i \equiv \{ \text{ the algorithm outputs cluster } V_i\}$. Let $$J_2 = \{i \in \{1,\dots, n\}: P_0(\cE_i) \le \frac{10}k\}.$$ Moreover, since $\sum_{i=1}^k P_0(\cE_i) \le 1$ we must have,
$
(k -|J_2|)\frac{10}k \le 1,
$
or $|J_2| \ge \frac{9k}{10}$. Therefore, $J = J_1 \cap J_2$ has size,
$$
|J| \ge 2 \cdot \frac{9k}{10} - k = \frac{4k}{5}.
$$

Now let us assume that, we are given a vertex $v \in V_j$ for some $j \in J$ to cluster. The probability of correct clustering is $P_j(\cE_j)$. We must have,
\begin{align*}
P_j(\cE_j) &= P_0(\cE_j) + P_j(\cE_j) -P_0(\cE_j)
 \le  \frac{10}k + | P_0(\cE_j) -P_j(\cE_j)|\\
 &\quad \le \frac{10}k + \| P_0 -P_j\|_{TV}
 \le  \frac{10}k + \sqrt{ \frac{1}{2}D(P_0\|P_j)}.
 \end{align*}
where we again used the definition of the total variation distance and  in the last step we have used the Pinsker's inequality (lemma \ref{lem:pinsker}).
The task is now to bound the divergence $D(P_0 \|P_j)$. Recall that $P_0$ and $P_j$ are the joint distributions of the 
independent random variables (answers to queries) that are identical to one of two Bernoulli random variables:$Y$, which is Bernoulli($p$), or
$Z$, which is Bernoulli($1-p$). Let $X_1,\dots,X_Q$ denote the outputs of the queries, all independent random variables. 
We must have, from the chain rule (lemma \ref{lem:chain}),
\begin{align*}
D(P_0 \|P_j) &= \sum_{i=1}^Q 
D(P_0(x_i|x_1,\dots, x_{i-1}) \| P_j(x_i|x_1,\dots, x_{i-1})) \\
&=\sum_{i=1}^Q \sum_{(x_1,\dots, x_{i-1}) \in \{0,1\}^{i-1}}P_0(x_1,\dots, x_{i-1})D(P_0(x_i|x_1,\dots, x_{i-1}) \| P_j(x_i|x_1,\dots, x_{i-1})).
\end{align*}
Note that, for the random variable $X_i$, the term 
$D(P_0(x_i|x_1,\dots, x_{i-1}) \| P_j(x_i|x_1,\dots, x_{i-1}))$
will contribute to $D(p\|1-p)$ only when the query involves the cluster $V_j$. Otherwise the term will contribute to $0$.  Hence, 
\begin{align*}
D(P_0 \|P_j)& = \sum_{i=1}^Q \sum_{(x_1,\dots, x_{i-1}) \in \{0,1\}^{i-1}: i \text{th query involves } V_j}P_0(x_1,\dots, x_{i-1})D(p\|1-p)\\
& = D(p\|1-p)\sum_{i=1}^Q \sum_{(x_1,\dots, x_{i-1}) \in \{0,1\}^{i-1}: i \text{th query involves } V_j}P_0(x_1,\dots, x_{i-1})\\
& = D(p\|1-p)\sum_{i=1}^Q P_0( i \text{th query involves } V_j)
 = D(p\|1-p) \avg_0 Q_j  \le \frac{10Q}{k}D(p\|1-p).
\end{align*}
Now plugging this in,
  \begin{align*}
 D(P_0 \|P_j) 
 \le  \frac{10}k + \sqrt{ \frac{1}{2}\frac{10Q}{k}D(p\|1-p)}
 \le  \frac{10}k + \sqrt{ \frac{1}{2}},
\end{align*}
if $Q \le \frac{k}{10D(p\|1-p)}$. On the other hand, when $p =0$, $P_j(\cE_j) <1$ when  $\avg_0 Q_j < 1$. Therefore $\frac{10Q}{k} \ge 1$
whenever $p=0$.

\end{proof}

Now we are ready to prove Theorem \ref{thm:faulty}.

\begin{proof}[Proof of Theorem \ref{thm:faulty}]
We will show that claim by considering {\em any} input, with a restriction on either the maximum or the minimum cluster size.  We consider the following two cases for the proof.

\vspace{0.1in}
\noindent{\em Case 1: the maximum size of a cluster is $\le \frac{4n}{k}$.}

\vspace{0.1in}
Suppose, total number of queries is = $T$.  That means number of vertices involved in the queries is $\le 2T$. Note that,  there are $k$ clusters and $n$ elements.

Let $U$ be the set of vertices that are involved in less than $\frac{16T}{n}$ queries. Clearly,
$$
(n-|U|)\frac{16T}{n} \le 2T,
\quad \text{ or } |U| \ge \frac{7n}8.$$ 

Now we know from Lemma \ref{lem:clus} that there exists $\frac{4k}{5}$ clusters such that  a vertex $v$ from any one of these clusters will be assigned to a wrong cluster by any randomized algorithm with positive probability unless the expected number  of queries involving this vertex is more than $\frac{k}{10D(p\|1-p)}$, for $p >0$, and $\frac{k}{10}$ when $p =0$.

We claim that $U$ must have an intersection with at least one of these $\frac{4k}{5}$ clusters. If not, then  more than $\frac{7n}8$ vertices must belong to less than $k -\frac{4k}{5} = \frac{k}{5}$ clusters. Or the maximum size of a cluster will be $\frac{7n\cdot 5}{8k} >\frac{4n}{k},$
which is prohibited according to our assumption.

Consider the case, $p >0$. Now each vertex in the intersection of $U$ and the  $\frac{4k}{5}$ clusters are going to be assigned to an incorrect cluster with positive probability
if,
$
\frac{16T}{n} \le \frac{k}{10D(p\|1-p)}.
$
Therefore we must have 
$$
T \ge  \frac{nk}{160D(p\|1-p)}.
$$

Similarly, when $p =0$ we must have,
$
T \ge  \frac{nk}{160}.
$

\vspace{0.2in}
\noindent{\em Case 2: the minimum size of a cluster is $\ge \frac{n}{20k}$.}

Let $U'$ be the set of clusters that are involved in at most $\frac{16T}{k}$ queries. That means,
$
(k-|U'|)\frac{16T}{k} \le 2T.
$
This implies, $ |U'| \ge \frac{7k}{8}$. 

Now we know from lemma \ref{lem:clus} that there exists $\frac{4k}{5}$ clusters (say $U^\ast$) such that  a vertex $v$ from any one of these clusters will be assigned to a wrong cluster by any randomized algorithm with positive probability unless the expected number  of queries involving this vertex is more than $\frac{k}{10D(p\|1-p)}$, $p >0$ and $\frac{k}{10}$ for $p =0$.

Quite clearly $|U^\ast \cap U| \ge \frac{7k}{8} + \frac{4k}{5} - k = \frac{27k}{40}$.

Consider a cluster $V_i$ such that $i \in U^\ast \cap U$, which is always possible because the intersection is nonempty.
$V_i$ 
  is involved in at most $\frac{16T}{k}$ queries. Let the minimum size of any cluster be $t$. Now, at least half of the vertices of $V_i$ must each be involved in at most $\frac{32T}{kt}$ queries. Now each of these vertices must be involved in at least $\frac{k}{10D(p\|1-p)}$ queries (see  Lemma \ref{lem:clus}) to avoid being assigned to a wrong cluster with positive probability (for the case of $p =0$ this number would be $\frac{k}{10}$).

This means, $$\frac{32T}{kt} \ge \frac{k}{10D(p\|1-p)},
\qquad \text{ or }  \quad T = \Omega\Big(\frac{nk}{D(p\|1-p)}\Big),
$$
 for $p>0$,
since $t \ge \frac{n}{20k}$. Similarly when $p =0$ we need $ T = \Omega(nk)$.
\end{proof}

\subsubsection{Upper Bound}\label{sec:faultyub}
Now we provide an algorithm to retrieve the clustering with the help of the faulty oracle when no side information is present.
The algorithm is summarized in Algorithm \ref{algo:cc-error-noside}.
The algorithm works as follows. It maintains an active list of clusters $A$, and a sample graph $G'=(V',E')$ which is an induced subgraph of $G$. Initially, both of them are empty. The algorithm always maintains the invariant that any cluster in $A$ has at least $c\log{n}$ members where $c=\frac{6}{\lambda^2}$, and $p=\frac{1}{2}-\lambda$. Note that the algorithm knows $\lambda$. Furthermore, all $V'(G') \times V'(G')$ queries have been made. Now, when a vertex $v$ is considered by the algorithm (step $3$), first we check if $v$ can be included in any of the clusters in $A$. This is done by picking $c\log{n}$ distinct members from each cluster, and querying $v$ with them. If majority of these questions return $+1$, then $v$ is included in that cluster, and we proceed to the next vertex. Otherwise, if $v$ cannot be included in any of the clusters in $A$, then we add it to $V'(G')$, and ask all possible queries to the rest of the vertices in $G'$ with $v$. Once $G'$ has been modified, we extract the heaviest weight subgraph from $G'$ where weight on an edge $(u,v) \in E(G')$ is defined as $\omega_{u,v}=+1$ if the query answer for that edge is $+1$ and $-1$ otherwise. If that subgraph contains $c\log{n}$ members then we include it as a cluster in $A$. At that time, we also check whether any other vertex $u$ in $G'$ can join this newly formed cluster by counting if the majority of the (already) queried edges to this new cluster gave answer $+1$. At the end, all the clusters in $A$, and the maximum likelihood clustering from $G'$ is returned.  

Before showing the correctness of Algorithm \ref{algo:cc-error-noside}, we elaborate on finding the maximum likelihood estimate for the clusters in $G$.

\paragraph*{Finding the Maximum Likelihood Clustering of $G$ with faulty oracle}

We have an undirected graph $G(V\equiv[n],E)$, such that $G$ is a union of $k$ disjoint cliques $G_i(V_i, E_i)$, $i =1, \dots, k$. 
 The subsets $V_i \in [n]$ are unknown to us.
 The adjacency matrix of $G$ is a block-diagonal matrix. Let us denote this matrix by $A = (a_{i,j})$.

Now suppose, each edge  of $G$ is erased  independently with probability $p$, and at the same time each non-edge is
replaced with an edge with probability $p$. Let the resultant adjacency matrix of the modified graph be 
$Z= (z_{i,j})$. The aim is to recover $A$ from $Z$. 

The maximum likelihood recovery is given by the following:
\begin{align*}
\max_{S_\ell, \ell = 1, \dots: V = \sqcup_{\ell} S_\ell}& \prod_{\ell} \prod_{i, j \in S_\ell, i \ne j}P_+(z_{i,j}) \prod_{r,t, r\ne t} \prod_{i \in S_r, j\in S_t} P_-(z_{i,j})\\
= \max_{S_\ell, \ell = 1, \dots: V = \sqcup_{\ell=1} S_\ell}&\prod_{\ell} \prod_{i, j \in S_\ell, i \ne j}\frac{P_+(z_{i,j})}{P_-(z_{i,j})}
\prod_{i,j \in V, i \ne j} P_-(z_{i,j}).
\end{align*}
where, $P_+(1) =1-p, P_+(0) =p, P_-(1) =p, P_-(0) = 1-p.$
Hence, the ML recovery asks for, 
$$
\max_{S_\ell, \ell = 1, \dots: V = \sqcup_{\ell=1} S_\ell}\sum_{\ell} \sum_{i, j \in S_\ell, i \ne j}\ln \frac{P_+(z_{i,j})}{P_-(z_{i,j})}.
$$
Note that, $$\ln \frac{P_+(0)}{P_-(0)} = - \ln \frac{P_+(1)}{P_-(1)} = \ln\frac{p}{1-p} .$$
Hence the ML estimation is,
\begin{align}
\label{eq:ml}
\max_{S_\ell, \ell = 1, \dots: V = \sqcup_{\ell=1} S_\ell}\sum_{\ell} \sum_{i, j \in S_\ell, i \ne j}\omega_{i,j},
\end{align}
where $\omega_{i,j} = 2z_{i,j} -1, i \ne j$, i.e., $\omega_{i,j} =1,$ when $z_{i,j} =1$ and $\omega_{i,j} =-1$ when $z_{i,j} =0, i \ne j$. Further $\omega_{i,i} = z_{i,i} =0, i = 1, \dots, n.$

 Note that \eqref{eq:ml} is equivalent to finding correlation clustering in $G$ with the objective of maximizing the consistency with the edge labels, that is we want to maximize the total number of positive intra-cluster edges and total number of negative inter-cluster edges \cite{bbc:04,ms:10,mmv:14}. This can be seen as follows.
 
  \begin{align*}
& \max_{S_\ell, \ell = 1, \dots: V = \sqcup_{\ell=1} S_\ell}\sum_{\ell} \sum_{i, j \in S_\ell, i \ne j}\omega_{i,j}\\
&\equiv \max_{S_\ell, \ell = 1, \dots: V = \sqcup_{\ell=1} S_\ell} \big[\sum_{\ell} \sum_{i, j \in S_\ell, i \ne j}\big|(i,j): \omega_{i,j}=+1\big|-\big|(i,j): \omega_{i,j}=-1\big|\big]+\sum_{i,j \in V, i \ne j}\big|(i,j): \omega_{i,j}=-1\big| \\
&=\max_{S_\ell, \ell = 1, \dots: V = \sqcup_{\ell=1} S_\ell} \big[\sum_{\ell} \sum_{i, j \in S_\ell, i \ne j}\big|(i,j): \omega_{i,j}=+1\big|+\big[\sum_{r,t: r \ne t} \big|(i,j): i \in S_r, j \in S_t, \omega_{i,j}=-1\big|\big].
\end{align*}
Therefore \eqref{eq:ml} is same as correlation clustering, however viewing it as obtaining clusters with maximum intra-cluster weight helps us to obtain the desired running time of our algorithm. Also, note that, we have a random instance of correlation clustering here, and not a worst case instance.

We are now ready to prove the correctness of Algorithm \ref{algo:cc-error-noside}.

\begin{algorithm}
\caption{\cc~with Error \& No Side Information. Input: $\{V\}$}
\label{algo:cc-error-noside}
\begin{algorithmic}[1]
\State $V'=\emptyset, E'=\emptyset, G'=(V',E')$
\State $A=\emptyset$
\While{$\exists v \in V$ yet to be clustered}
\For{ each cluster $\calC \in A$}
\LineComment{Set $c=\frac{6}{\lambda^2}$ where $\lambda\equiv \frac{1}{2}-p$.}
\State Select $u_1, u_2,.., u_l$, where $l=c\log{n}$, distinct members from $\calC$ and obtain $\mathcal{O}_{p}(u_i,v)$, $i=1,2,..,l$. If the majority of these queries return $+$, then include $v$ in $\calC$. Break.
\EndFor
\If{ $v$ is not included in any cluster in $A$}
\State Add $v$ to $V'$. For every $u \in V' \setminus v$, obtain $\mathcal{O}_{p}(v,u)$. Add an edge $(v,u)$ to $E'(G')$ with weight $\omega_{u,v}=+1$ if $\mathcal{O}_{p}(u,v)==+1$, else with $\omega_{u,v}=-1$
\State \label{eq:find_set} Find the heaviest weight subgraph $S$ in $G'$. If $|S| \geq c\log{n}$, then add $S$ to the list of clusters in $A$, and remove the incident vertices and edges on $S$ from $V',E'$.
\While{ $\exists z \in V'$ with $\sum_{u \in S} \omega_{z,u} > 0$}
\State Include $z$ in $S$ and remove $z$ and all edges incident on it from $V',E'$.
\EndWhile
\EndIf
\EndWhile\\
\Return all the clusters formed in $A$ and the ML estimates from $G'$
\end{algorithmic}
\end{algorithm}

\paragraph*{Correctness of Algorithm \ref{algo:cc-error-noside}}
To establish the correctness of Algorithm \ref{algo:cc-error-noside}, we show the following. Suppose all $\binom{n}{2}$ queries on $V \times V$  have been made. If the Maximum Likelihood (ML) estimate of $G$ with these $\binom{n}{2}$  answers is same as the true clustering of $G$, then Algorithm \ref{algo:cc-error-noside} finds the true clustering with high probability. There are few steps to prove the correctness. 

The first step is to show that any set $S$ that is retrieved in step \ref{eq:find_set} of Algorithm \ref{algo:cc-error-noside} from $G'$, and added to $A$ is a subcluster of $G$ (Lemma \ref{lemma:mlG'}). This establishes that all clusters in $A$ at any time are subclusters of some original cluster in $G$. Next, we show that vertices that are added to a cluster in $A$, are added correctly, and no two clusters in $A$ can be merged (Lemma \ref{lemma:vertex}). Therefore, clusters obtained from $A$, are the true clusters. Finally, the remaining of the clusters can be retrieved from $G'$ by computing a ML estimate on $G'$ in step $15$, leading to theorem \ref{lemma:correct}.

\begin{lemma}
\label{lemma:mlG'}
Let $c'=6c=\frac{36}{\lambda^2}$, where $\lambda=\frac{1}{2}-p$. Algorithm \ref{algo:cc-error-noside} in step \ref{eq:find_set} returns a subcluster of $G$ of size at least $c\log{n}$ with high probability if $G'$ contains a subcluster of $G$ of size at least $c'\log{n}$. Moreover, Algorithm \ref{algo:cc-error-noside} in step \ref{eq:find_set} does not return any set of vertices of size at least $c\log{n}$ if $G'$ does not contain a subcluster of $G$ of size at least $c\log{n}$.
\end{lemma}
\begin{proof}
Let $V'=\bigcup V'_i$, $i\in [1,k]$,  $V'_i \cap V'_j =\emptyset$ for $i \neq j$, and $V'_i \subseteq V_i(G)$. Suppose without loss of generality $|V'_1| \geq |V'_2| \geq ....\geq |V'_k|$.

The lemma is proved via a series of claims.
\begin{claim}
\label{claim:1}
Let $|V'_1| \geq c'\log{n}$. Then in step \ref{eq:find_set}, a set $S \subseteq V_i$ for some $i \in [1,k]$ will be returned with size at least $c\log{n}$ with high probability.
\end{claim}

For an $i: |V'_i| \ge c' \log n,$ we have
\begin{align*}
\avg \sum_{s, t \in V'_i, s<t}\omega_{s,t} = \binom{|V'_i|}{2}((1-p)-p) = (1-2p)\binom{|V'_i|}{2}.
\end{align*}
Since $\omega_{s,t}$ are independent binary random variables, using the Hoeffding's inequality (Lemma \ref{lem:hoef1}),
\begin{align*}
\Pr\Big( \sum_{s, t \in V'_i, s<t}\omega_{s,t} \le \avg \sum_{s, t \in V'_i, s<t}\omega_{s,t} - u \Big) \le e^{-\frac{ u^2 }{2\binom{|V'_i|}{2}}}.
\end{align*}
Hence,
\begin{align*}
\Pr\Big( \sum_{s, t \in V'_i, s<t}\omega_{s,t} >(1-\delta) \avg \sum_{s, t \in V'_i, s<t}\omega_{s,t} \Big) \ge 1 - e^{-\frac{ \delta^2(1-2p)^2 \binom{|V'_i|}{2} }{2}}.
\end{align*}
Therefore with high probability $\sum_{s, t \in V'_i, s<t}\omega_{s,t} > (1-\delta) (1-2p)\binom{|V'_i|}{2} \ge  (1-\delta) (1-2p)\binom{c' \log n}2 >
\frac{c'^2}{3}(1-2p) \log^2 n,$ for an appropriately chosen $\delta$ (say $\delta=\frac{1}{3}$).

So, Algorithm \ref{algo:cc-error-noside} in step \eqref{eq:find_set} must return a set $S$ such that $|S| \ge c' \sqrt{\frac{2(1-2p)}{3}} \log n=c''\log{n}$ (define $c''=c' \sqrt{\frac{2(1-2p)}{3}}$) with high probability -  since otherwise $$\sum_{i, j \in S, i < j}\omega_{i,j} < \binom{c' \sqrt{\frac{2(1-2p)}{3}} \log n}{2} < \frac{c'^2}{3}(1-2p) \log^2 n.$$

Now let $S \nsubseteq V_i$ for any $i$. Then $S$ must have intersection with at least $2$ clusters. Let $V_i \cap S = C_i$ and
let $j^\ast = \arg \min_{i: C_i \neq \emptyset} |C_i|$. We claim that,
\begin{equation}\label{eq:reduc}
\sum_{i, j \in S, i < j}\omega_{i,j}  < \sum_{i, j \in S \setminus C_{j^\ast}, i < j}\omega_{i,j},
\end{equation}
with high probability.
Condition \eqref{eq:reduc} is equivalent to,
$$
\sum_{i, j \in  C_{j^\ast}, i < j}\omega_{i,j} + \sum_{i \in C_{j^\ast}, j \in S \setminus C_{j^\ast}} \omega_{i,j} <0.
$$
However this is true because,
\begin{enumerate}
\item $
\avg \Big(\sum_{i, j \in  C_{j^\ast}, i < j}\omega_{i,j} \Big) = (1-2p) \binom{|C_{j^\ast}|}{2}$ and 
$\avg\Big(\sum_{i \in C_{j^\ast}, j \in S \setminus C_{j^\ast}} \omega_{i,j} \Big) = - (1-2p)|C_{j^\ast}|\cdot |S\setminus C_{j^\ast}|.
$
\item As long as $|C_{j^\ast}| \ge 2\sqrt{\log n}$ we have, from Hoeffding's inequality (Lemma \ref{lem:hoef1}),
$$
\Pr\Big(\sum_{i, j \in  C_{j^\ast}, i < j}\omega_{i,j}  \ge (1+\lambda) (1-2p) \binom{|C_{j^\ast}|}{2}\Big) \le e^{-\frac{\lambda^2(1-2p)^2\binom{|C_{j^\ast}|}{2}}{2}} = o_n(1).
$$ 
While at the same time, 
$$
\Pr\Big( \sum_{i \in C_{j^\ast}, j \in S \setminus C_{j^\ast}} \omega_{i,j}  \ge - (1-\lambda) (1-2p)|C_{j^\ast}|\cdot |S\setminus C_{j^\ast}|\Big) \le e^{-\frac{\lambda^2 (1-2p)^2 |C_{j^\ast}|\cdot |S\setminus C_{j^\ast}|}{2}} = o_n(1).
$$
In this case of course with high probability $$
\sum_{i, j \in  C_{j^\ast}, i < j}\omega_{i,j} + \sum_{i \in C_{j^\ast}, j \in S \setminus C_{j^\ast}} \omega_{i,j} <0.
$$
\item When  $|C_{j^\ast}| < 2\sqrt{\log n}$, we have,
$$
\sum_{i, j \in  C_{j^\ast}, i < j}\omega_{i,j} \le \binom{|C_{j^\ast}|}{2} \le 2 \log^2 n.
$$
While at the same time, 
$$
\Pr\Big( \sum_{i \in C_{j^\ast}, j \in S \setminus C_{j^\ast}} \omega_{i,j}  \le (1-\lambda) (1-2p)|C_{j^\ast}|\cdot |S\setminus C_{j^\ast}|\Big) \le e^{-\frac{\lambda^2 (1-2p)^2 |C_{j^\ast}|\cdot |S\setminus C_{j^\ast}|}{2}} = o_n(1).
$$
Hence, even in this case, with high probability,
$$
\sum_{i, j \in  C_{j^\ast}, i < j}\omega_{i,j} + \sum_{i \in C_{j^\ast}, j \in S \setminus C_{j^\ast}} \omega_{i,j} <0.
$$
\end{enumerate}
Hence \eqref{eq:reduc} is true with high probability. But then the algorithm \ref{algo:cc-error-noside} in step \ref{eq:find_set} would not return $S$, but will return $S \setminus C_{j^\ast}$. Hence, we have run into a contradiction. This means $S \subseteq V_i$ for some $V_i$. 

We know $|S| \ge c' \sqrt{\frac{2(1-2p)}{3}} \log n$, while $|V_1'| \geq c'\log{n}$. In fact, with high probability, $|S| \geq \frac{(1-\delta)}{2}c'\log{n}$. Since all the vertices in $S$ belong to the same cluster in $G$, this holds again by the application of Hoeffding's inequality. Otherwise, the probability that the weight of $S$ is at least as high as the weight of $V_1'$ is at most $\frac{1}{n^2}$.

\begin{claim}
\label{claim:2}
If $|V'_1| < c\log{n}$. then in step \ref{eq:find_set} of Algorithm \ref{algo:cc-error-noside}, no subset of size $> c\log{n}$ will be returned.  
\end{claim}

If Algorithm \ref{algo:cc-error-noside} in step \ref{eq:find_set} returns a set $S$ with $|S| > c\log{n}$ then $S$ must have intersection with at least $2$ clusters in $G$. Now following the same argument as in Claim \ref{claim:1} to establish Eq. \eqref{eq:reduc}, we arrive to a contradiction, and $S$ cannot be returned.

This establishes the lemma.
\end{proof}

\begin{lemma}
\label{lemma:vertex}
 The collection $A$ contains all the true clusters of $G$ of size $\geq c'\log{n}$ at the end of Algorithm \ref{algo:cc-error-noside} with high probability.
\end{lemma}
\begin{proof}
From Lemma \ref{lemma:mlG'}, any cluster that is computed in step \ref{eq:find_set} and added to $A$ is a subset of some original cluster in $G$, and has size at least $c\log{n}$ with high probability. Moreover, whenever $G'$ contains a subcluster of $G$ of size at least $c'\log{n}$, it is retrieved by our Algorithm and added to $A$.

A vertex $v$ is added to a cluster in $A$ either is step 5 or step 11. Suppose, $v$ has been added to some cluster $\calC \in A$. Then in both the cases, $|\calC| \geq c\log{n}$ at the time $v$ is added, and there exist $l=c\log{n}$ distinct members of $\calC$, say, $u_1,u_2,..,u_l$ such that majority of the queries of $v$ with these vertices returned $+1$. By the standard Chernoff-Hoeffding bound (Lemma \ref{lem:hoef1}), $\Pr(v \notin \calC) \leq \text{exp}(-c\log{n}\frac{(1-2p)^2}{12p})=\text{exp}(-c\log{n}\frac{2\lambda^2}{3(1+2\lambda)})\leq \text{exp}(-c\log{n}\frac{\lambda^2}{3})$, where the last inequality followed since $\lambda < \frac{1}{2}$. On the other hand, if there exists a cluster $\calC \in A$ such that $v \in \calC$, and $v$ has already been considered by the algorithm, then either in step 5 or step 11, $v$ will be added to $\calC$. This again follows by the Chernoff-Hoeffding bound, as $\Pr(v \text{ not included in } \calC \mid v \in \calC) \leq \text{exp}(-c\log{n}\frac{(1-2p)^2}{8(1-p)})=\text{exp}(-c\log{n}\frac{\lambda^2}{1+2\lambda})\leq \text{exp}(-c\log{n}\frac{\lambda^2}{2})$. Therefore, if we set $c=\frac{6}{\lambda^2}$, then for all $v$, if $v$ is included in a cluster in $A$, the assignment is correct with probability at least $1-\frac{2}{n}$. Also, the assignment happens as soon as such a cluster is formed in $A$.

Furthermore, two clusters in $A$ cannot be merged. Suppose, if possible there are two clusters $\calC_1$ and $\calC_2$ both of which are proper subset of some original cluster in $G$. Let without loss of generality $\calC_2$ is added later in $A$. Consider the first vertex $v \in \calC_2$ that is considered by our Algorithm \ref{algo:cc-error-noside} in step $3$. If $\calC_1$ is already there in $A$ at that time, then with high probability $v$ will be added to $\calC_1$ in step 5. Therefore, $\calC_1$ must have been added to $A$ after $v$ has been considered by our algorithm and added to $G'$. Now, at the time $\calC_1$ is added to $A$ in step 9, $v \in V'$, and again $v$ will be added to $\calC_1$ with high probability in step 11--thereby giving a contradiction.

This completes the proof of the lemma.
\end{proof}
All this leads us to the following theorem. 
\begin{theorem}
\label{lemma:correct}
If the ML estimate on $G$ with all possible $\binom{n}{2}$ queries return the true clustering, then Algorithm \ref{algo:cc-error-noside} returns the true clusters with high probability. Moreover, Algorithm \ref{algo:cc-error-noside} returns all the true clusters of $G$ of size at least $c'\log{n}$ with high probability.
\end{theorem}
\begin{proof}
From Lemma \ref{lemma:mlG'} and Lemma \ref{lemma:vertex}, $A$ contains all the true clusters of $G$ of size at least $c'\log{n}$ with high probability. Any vertex that is not included in the clusters in $A$ at the end of Algorithm \ref{algo:cc-error-noside} are in $G'$, and $G'$ contains all possible pairwise queries among them. Clearly, then the ML estimate of $G'$ will be the true ML estimate of $G$ restricted to these clusters.
\end{proof}

\paragraph*{Query Complexity of Algorithm \ref{algo:cc-error-noside}}
\begin{lemma}
Let $p = \frac12 -\lambda$. The query complexity of Algorithm \ref{algo:cc-error-noside} is $\frac{36nk\log{n}}{\lambda^2}$.
\end{lemma}
\begin{proof}
Let there be $k'$ clusters in $A$ when $v$ is considered in step $3$ of Algorithm \ref{algo:cc-error-noside}. Then $v$ is queried with at most $ck'\log{n}$ current members, $c\log{n}$ each from these $k'$ clusters. If the membership of $v$ does not get determined then $v$ is queried with all the vertices in $G'$. We have seen in the correctness proof (Lemma \ref{lemma:correct}) that if $G'$ contains at least $c'\log{n}$ vertices from any original cluster, then ML estimate on $G'$ retrieves those vertices as a cluster in step 9 with high probability. Hence, when $v$ is queried with all vertices in $G'$, $|V'|\leq (k-k')c'\log{n}$. Thus the total number of queries made to determine the membership of $v$ is at most $c'k\log{n}$, where $c'=6c=\frac{36}{\lambda^2}$ when the error probability $p=\frac{1}{2}-\lambda$. This gives the query complexity of Algorithm \ref{algo:cc-error-noside} considering all the vertices.

This matches the lower bound computed in Section \ref{sec:error-lc} within a $\log{n}$ factor, since $D(p\|1-p) = (1-2p) \ln \frac{1-p}{p} = 2\lambda\ln\frac{1/2+\lambda}{1/2 -\lambda} =2\lambda\ln(1+\frac{2\lambda}{1/2-\lambda}) \le \frac{4\lambda^2}{1/2-\lambda} = O(\lambda^2)$.
\end{proof}

Now combining all these we get the statement of Theorem \ref{theorem:cc-error}.

\begin{theorem-n}[\ref{theorem:cc-error}]
Faulty Oracle with No Side Information. There exists an algorithm with query complexity $O(\frac{1}{\lambda^2}nk\log{n})$ for \cc~that returns $\hat{G}$, ML estimate of $G$ with all $\binom{n}2$ queries, with high probability when query answers are incorrect with probability $p=\frac{1}{2}-\lambda$. Noting that, $D(p\|1-p) \le \frac{4\lambda^2}{1/2-\lambda}$, this matches the information theoretic lower bound on the query complexity within a $\log{n}$ factor. Moreover, the algorithm returns all the true clusters of $G$ of size at least $\frac{36}{\lambda^2}\log{n}$ with high probability. 
\end{theorem-n}

\paragraph*{Running Time of Algorithm \ref{algo:cc-error-noside} and Further Discussions}

In step \ref{eq:find_set} of Algorithm \ref{algo:cc-error-noside}, we need to find a large cluster of size at least $O(\frac{1}{\lambda^2}\log{n})$ of the original input $G$ from $G'$. By Lemma \ref{lemma:mlG'}, if we can extract the heaviest weight subgraph in $G'$ where edges are labelled $\pm1$, and that subgraph meets the required size bound, then with high probability, it is a subset of an original cluster. This subset can of course be computed in $O(n^{\frac{1}{\lambda^2}\log{n}})$ time. Since size of $G'$ is bounded by $O(\frac{k}{\lambda^2}\log{n})$, the running time is $O([\frac{k}{\lambda^2}\log{n}]^{\frac{1}{\lambda^2}\log{n}})$. While, query complexity is independent of running time, it is unlikely that this running time can be improved to a polynomial. This follows from the planted clique conjecture.

\begin{conjecture}[Planted Clique Hardness]
Given an Erd\H{o}s-R\'{e}nyi random graph $G(n,p)$, with $p=\frac{1}{2}$, the planted clique conjecture states that if we plant in $G(n,p)$ a clique of size $t$ where $t=[O(\log{n}), o(\sqrt{n})]$, then there exists no polynomial time algorithm to recover the largest clique in this planted model.
\end{conjecture}

Given such a graph with a planted clique of size $t=\Theta(\log{n})$, we can construct a new graph $H$ by randomly deleting each edge with probability $\frac{1}{3}$. Then in $H$, there is one cluster of size $t$ where edge error probability is $\frac{1}{3}$ and the remaining clusters are singleton with inter-cluster edge error probability being $(1-\frac{1}{2}-\frac{1}{6})=\frac{1}{3}$. So, if we can detect the heaviest weight subgraph in polynomial time in Algorithm \ref{algo:cc-error-noside}, there will be a polynomial time algorithm for the planted clique problem.

 \paragraph*{Polynomial time algorithm}
 We can reduce the running time from quasi-polynomial to polynomial, by paying higher in the query-complexity. Suppose, we accept a subgraph extracted from $G'$ as valid and add it to $A$ iff its size is $\Omega(k)$. Then note that since $G'$ can contain at most $k^2$ vertices, such a subgraph can be obtained in polynomial time following the algorithm of correlation clustering with noisy input \cite{ms:10}, where all the clusters of size at least $O(\sqrt{n})$ are recovered on a $n$-vertex graph. Since our ML estimate is correlation clustering, we can employ \cite{ms:10}. For $k \geq \frac{1}{\lambda^2}\log{n}$, the entire analysis remains valid, and we get a query complexity of $\tilde{O}(nk^2)$ as opposed to $O(\frac{nk}{\lambda^2})$. If $k < \frac{1}{\lambda^2}\log{n}$, then clusters that have size less than $ \frac{1}{\lambda^2}\log{n}$ are anyway not recoverable. Note that, any cluster that has size less than $k$ are not recovered in this process, and this bound only makes sense when $k < \sqrt{n}$. When $k\geq\sqrt{n}$, we can however recover all clusters of size at least $O(\sqrt{n})$.
 
\begin{corollaryn}[\ref{cor:er-poly}]
 There exists a polynomial time algorithm with query complexity $O(\frac{1}{\lambda^2}nk^2)$ for \cc~when query answers are incorrect with probability $\frac{1}{2}-\lambda$, which recovers all clusters of size at least $O(\max{ \{\frac{1}{\lambda^2}\log{n},k\}})$ in $G$.
 \end{corollaryn}
 
 This also leads to an improved algorithm for correlation clustering over noisy graph. Previously, the works of \cite{ms:10,bbc:04} can only recover cluster of size at least $O(\sqrt{n})$. However, now if $k \in [\Omega(\frac{\log{n}}{\lambda^2}), o(\sqrt{n})]$, using this algorithm, we can recover all clusters of size at least $k$.

 \subsubsection{With Side Information}\label{sec:faultysideub}
 The algorithm for \cc~with side information when crowd may return erroneous answers is a direct combination of Algorithm \ref{algo:cc-exact} and Algorithm \ref{algo:cc-error-noside}. We assume side information is less accurate than querying because otherwise, querying is not useful. Or in other words $\Delta(f_g,f_r) < \Delta(p,1-p)$.
 
 We therefore use only the queried answers to extract the heaviest subgraph from $G'$, and add that to the list $A$. For the clusters in list $A$, we follow the strategy of Algorithm \ref{algo:cc-exact} to recover the underlying clusters. The pseudocode is given in Algorithm \ref{algo:cc-error-side}. The correctness of the algorithm follows directly from the analysis of Algorithm \ref{algo:cc-exact} and  Algorithm \ref{algo:cc-error-noside}.  
 
 We now analyze the query complexity. Consider a vertex $v$ which needs to be included in a cluster. Let there be $(r-1)$ other vertices from the same cluster as $v$ that have been considered by the algorithm prior to $v$.
\begin{enumerate}
 \item Case 1. $r \in [1,c\log{n}]$, the number of queries is at most $kc\log{n}$. In that case $v$ is added to $G'$ according to Algorithm \ref{algo:cc-error-noside}.
 \item Case 2. $r \in (c\log{n},2M]$, the number of queries can be $k*c\log{n}$. In that case, the cluster that $v$ belongs to is in $A$, but has not grown to size $2M$. Recall $M=O(\frac{\log{n}}{\Delta(f_+,f_-)})$. In that case, according to Algorithm \ref{algo:cc-exact}, $v$ may need to be queried with each cluster in $A$, and according to Algorithm \ref{algo:cc-error-noside}, there can be at most $c\log{n}$ queries for each cluster in $A$. 
 \item Case 3. $r \in (2R,|C|]$, the number of queries is at most $c\log{n}*\log{n}$. In that case, according to Algorithm \ref{algo:cc-exact}, $v$ may need to be queried with at most $\lceil \log{n} \rceil$ clusters in $A$, and according to Algorithm \ref{algo:cc-error-noside}, there can be at most $c\log{n}$ queries for each chosen cluster in $A$. 
  \end{enumerate}
  
  Hence, the total number of queries per cluster is at most $O(kc^2(\log{n})^2+(2M-c\log{n})kc\log{n}+(|C|-2M)c(\log{n})^2)$. So, over all the clusters, the query complexity is $O(nc(\log{n})^2+k^2Mc\log{n})$.
Note that, if have instead insisted on a Monte Carlo algorithm with known $f_+$ and $f_-$, then the query complexity would have been $O(k^2Mc\log{n})$. Recall that $\Delta(p\|(1-p))=O(\lambda^2)$.

 \begin{theorem-n}[\ref{thm:div-new-err}]
 Let $f_+$ and $f_-$ be pmfs and $\min_i f_+(i), \min_i f_-(i) \ge \epsilon$ for a constant $\epsilon$.  With side information and faulty oracle with error probability $\frac{1}{2}-\lambda$, there exist an algorithm for \cc~with query complexity $O(\frac{k^2\log{n}}{\lambda^2\Delta(f_+,f_-)})$ with known $f_+$ and $f_-$, and an algorithm with expected query complexity $O(n+\frac{k^2\log{n}}{\lambda^2\Delta(f_+,f_-)})$ even when $f_+$ and $f_-$ are unknown that recover $\hat{G}$, ML estimate of $G$ with all $\binom{n}2$ queries with high probability.
 \end{theorem-n}

 \begin{algorithm}
\caption{\cc~with Error \& Side Information. Input: $\{V,W\}$}
\label{algo:cc-error-side}
\begin{algorithmic}[1]
\State $V'=\emptyset, E'=\emptyset, G'=(V',E')$, $A=\emptyset$
\While{$V\neq \emptyset$}
\State If $A$ is empty, then pick an arbitrary vertex $v$ and Go to Step \ref{eq:add_G'}
\LineComment{Let the number of current clusters in $A$ be $l \geq 1$}
\State Order the existing clusters in $A$ in nonincreasing size of current membership. 
\LineComment{Let $|\calC_{1}| \geq |\calC_{2}| \geq \ldots \geq |\calC_{l}|$ be the ordering (w.l.o.g).}
\For{$j=1$ to $l$ }
\State If $\exists v \in V$ such that $j=\max_{i\in[1,l]} {\sf Membership}(v, \cC_i)$, then select $v$ and Break;
\EndFor
 \State Select $u_1, u_2,.., u_l \in \calC_j$, where $l=c\log{n}$, distinct members from $\calC_j$ and obtain $\cO_p(u_i,v)$, $i=1,2,..,l$. $checked(v,j)=true$
\If{the majority of these queries return $+1$}
\State Include $v$ in $\calC_{j}$. $V=V \setminus v$
\Else
\LineComment{logarithmic search for membership in the large groups. Note $s \leq \lceil \log{k} \rceil$}
\State Group $\calC_1,\calC_2,...,\calC_{j-1}$ into $s$ consecutive classes $H_1, H_2,...,H_s$ such that the clusters in group $H_i$ have their current sizes in the range $[\frac{|\calC_1|}{2^{i-1}}, \frac{|\calC_1|}{2^i})$
\For{$i=1$ to $s$}
\State $j=\max_{a:\calC_a \in H_i}  {\sf Membership}(v, \cC_a)$
\State Select $u_1, u_2,.., u_l \in \calC_j$, where $l=c\log{n}$, distinct members from $\calC_j$ and obtain $\cO_p(u_i,v)$, $i=1,2,..,l$. $checked(v,j)=true$.
\If{the majority of these queries return $+1$}
\State Include $v$ in $\calC_{j}$. $V=V \setminus v$. Break.
\EndIf
\EndFor
\LineComment{exhaustive search for membership in the remaining groups in $A$}
\If{ $v \in V$}
\For{$i=1$ to $l+1$}
	\If{$i=l+1$} 
	\Comment{$v$ {\em does not belong to any of the existing clusters}}
	\State \label{eq:add_G'}Add $v$ to $V'$. Set $V=V\setminus v$
	\State For every $u \in V' \setminus v$, obtain $\cO_p(v,u)$. Add an edge $(v,u)$ to $E'(G')$ with weight $\omega(u,v)=+1$ if $\cO_p(v,u)==+1$, else with $\omega(u,v)=-1$
\State \label{eq1:find_set} Find the heaviest weight subgraph $S$ in $G'$. If $|S| \geq c\log{n}$, then add $S$ to the list of clusters in $A$, and remove the incident vertices and edges on $S$ from $V',E'$.
 \While{ $\exists z \in V'$ with $\sum_{u \in S} \omega(z,u) > 0$}
\State Include $z$ in $S$ and remove $z$ and all edges incident on it from $V',E'$.
\EndWhile
\State Break;
\Else
\If{ $checked(v,i)\neq true$ }
\State Select $u_1, u_2,.., u_l \in \calC_j$, where $l=c\log{n}$, distinct members from $\calC_j$ and $\cO_p(u_i,v)$, $i=1,2,..,l$. $checked(v,i)=true$.
\If{the majority of these queries return $+1$}
\State Include $v$ in $\calC_{j}$. $V=V \setminus v$. Break.
\EndIf
\EndIf
\EndIf
\EndFor
\EndIf
\EndIf
\EndWhile\\
\Return all the clusters formed in $A$ and the ML estimates from $G'$
\end{algorithmic}
\end{algorithm}

%% file: round-complexity.tex
 \section{Round Complexity}
 \label{sec:round}
 So far we have discussed developing algorithms for \cc~where queries are asked adaptively one by one. To use the crowd workers in the most efficient way, it is also important to incorporate as much parallelism as possible without affecting the query complexity by much. To formalize this, we allow at most $\Theta(n\log{n})$ queries simultaneously in a round, and then the goal is to minimize the number or rounds to recover the clusters. We show that the algorithms developed for optimizing query complexity naturally extends to the parallel version of minimizing the round complexity.

\subsection{\cc~with Perfect Oracle}\label{sec:roundpo}
 
 When crowd gives correct answers and there is no side information, then it is easy to get a round complexity of $k$ which is optimal within a $\log{n}$ factor as $\Omega(nk)$ is a lower bound on the query complexity in this case. One can just pick a vertex $v$, and then for every other vertex issue a query involving $v$. This grows the cluster containing $v$ completely. Thus in every round, one new cluster gets formed fully, resulting in a round complexity of $k$.
 
 We now explain the main steps of our algorithm when side information $W$ is available.
 
 \begin{enumerate}
 \item Sample $\sqrt{n\log{n}}$ vertices, and ask all possible $\binom{ \sqrt{n\log{n}}}{2}$ queries involving them. 
 
 \item Suppose $\calC_1,\calC_2,...,\calC_l$ are the clusters formed so far. Arrange these clusters in non-decreasing size of their current membership. For every vertex $v$ not yet clustered, choose the cluster $\calC_j$ with $j=\max_{i\in[1,l]} {\sf Membership}(v, \cC_i)$, and select at most $\lceil \log{n} \rceil$ clusters using steps $(11)$ and $(13)$ of Algorithm \ref{algo:cc-exact}.  Issue all of these at most $n\log{n}$ queries simultaneously, and based on the results, grow clusters $\calC_1,\calC_2,...,\calC_l$.
 
\item Among the vertices that have not been put into any cluster, pick  $\sqrt{n\log{n}}$ vertices uniformly at random, and ask all possible $\binom{ \sqrt{n\log{n}}}{2}$ queries involving them. Create clusters $\calC'_1,\calC'_2,...,\calC'_{l'}$ based on the query results. 
 
 \item Merge the clusters $\calC'_1,\calC'_2,...,\calC'_{l'}$ with $\calC_1,\calC_2,...,\calC_l$ by issuing a total of $ll'$ queries in $\lceil \frac{ll'}{n\log{n}}\rceil$ rounds.  Goto step 2.
  \end{enumerate} 

\paragraph*{Analysis} 
\label{sec:analysis-r1}
First, the algorithm computes the clusters correctly. Every vertex that is included in a cluster, is done so based on a query result. Moreover, no clusters in step $2$ can be merged. So all the clusters returned are correct.

We now analyzed the number of rounds required to compute the clusters.

In one iteration of the algorithm (steps $1$ to $4$), steps $1$ to $3$ each require one round, and issue at most $n\log{n}$ queries. Step $4$ requires at most $\frac{\min{(k^2,k\sqrt{n\log{n}})}}{n\log{n}}$ rounds and issue at most $\min{(k^2,k\sqrt{n\log{n}})}$ queries. This is because $l' \leq \sqrt{n\log{n}}$.

In step $2$, if $|\calC_i| \geq 2M$ (recall $M=O(\frac{\log{n}}{\Delta(f_+\|f_-)})$), for $i \in [1,l]$, then, at the end of that step, $\calC_i$ will be fully grown with high probability from the analysis of Algorithm \ref{algo:cc-exact}. This happens since with high probability any vertex that belongs to $\calC_i$ has been queried with some $u$ already in $\calC_i$. However, since we do not know $M$, we cannot identify whether $\calC_i$ has grown fully.

Consider the case when steps $1$ and $3$ have picked $6kM$ random vertices. Consider all those clusters that have size at least $\frac{n}{2k}$. Note that by Markov Inequality, at least $\frac{n}{2}$ vertices are contained in clusters of size at least $\frac{n}{2k}$. 

If we choose these $6kM$ vertices with replacement, then on expectation, the number of members chosen from each cluster of size $\frac{n}{2k}$ is $3M$, and with high probability above $2M$. This same concentration bound holds even though here sampling is done without replacement (Lemma \ref{lem:hoef2}). 

Therefore, after $6kM$ vertices have been chosen, and step $2$ has been performed, at least $\frac{n}{2}$ vertices get clustered and removed.

The number of iterations required to get $6kM$ random vertices is $\lceil \frac{6kM}{\sqrt{n\log{n}}} \rceil$. If $k^2 \geq n$, then the number of rounds required in each iteration is $2+\lceil \frac{k}{\sqrt{n\log{n}}}\rceil$. So the total number of rounds required to get $6kM$ vertices is $O(\frac{k^2M}{n\log{n}})$. And, finally to get all the vertices clustered, the number of rounds required will be $O(\frac{k^2M}{n})$, whereas the optimum round complexity could be $O(\frac{k^2M}{n\log{n}})$.

 If $k^2 < n$, then the number of rounds in each iteration is at most $3$. Hence the total number of iterations is at most $3+\frac{6kM}{\sqrt{n\log{n}}}$. If $kM \leq \sqrt{n\log{n}}$, then the number of rounds required is $O(1)$. Else, we have $kM > \sqrt{n\log{n}}$ and $k < \sqrt{n}$.  While our algorithm requires $O(\frac{kM}{\sqrt{n\log{n}}})$ rounds, we know the optimum round complexity is at least $O(\frac{k^2M}{n\log{n}})$. Overall, the gap may be at most $O(\frac{\sqrt{n\log{n}}}{k})=O(M)=O(\frac{\log{n}}{\Delta(f_+\|f_-)})$.
 
 This leads to Theorem \ref{thm:perfect-parallel}.
  \begin{theorem-n}[\ref{thm:perfect-parallel}]
Perfect Oracle with Side Information.  There exists an algorithm for \cc~with perfect oracle and unknown side information $f_+$ and $f_-$ such that it achieves a round complexity within $\tilde{O}(1)$ factor of the optimum when $k=\Omega(\sqrt{n})$ or $k=O(\frac{\sqrt{n}}{\Delta(f_+\|f_-)})$, and otherwise within $\tilde{O}({\Delta(f_+\|f_-)})$.  
  \end{theorem-n}

  
 \subsection{\cc~with Faulty Oracle}\label{sec:roundfo}
 We now move to the case of \cc~with faulty oracle. We obtain an algorithm with close to optimal round complexity when no side information is provided. By combining this algorithm with the one in the previous section, one can easily obtain an algorithm for \cc~with faulty oracle and side information. This is left as an exercise to the reader.
 
 We now give the algorithm for the case when crowd may return erroneous answer with probability $p=\frac{1}{2}-\lambda$ (known), and there is no side information.
 
 \begin{enumerate}
 \item Sample $\sqrt{n\log{n}}$ vertices uniformly at random, and ask all possible $\binom{ \sqrt{n\log{n}}}{2}$ queries involving them to form a subgraph $G''=(V'',E'')$
 \item Extract the highest weighted subgraph $S$ from $G''$ after setting a weight of $+1$ for every positive answer and $-1$ for every negative answer like in Algorithm \ref{algo:cc-error-noside}. If $|S| \geq c\log{n}$ where $c$ is set as in Algorithm \ref{algo:cc-error-noside}, then for every vertex not yet clustered issue $c\log{n}$ queries to distinct vertices in $S$ simultaneously in at most $c$ rounds. Grow $S$ by including any vertex where the majority of those queries returned is $+1$. Repeat step 2 as long as the extracted subgraph has size at least $c\log{n}$, else move to step $3$ while not all vertices have been clustered or included in $G''$.
 \item Among the vertices that have not been clustered yet, Pick  $r$ vertices $S_r$ uniformly at random, and ask all possible $\binom{ r}{2}+r|V''|$ queries among $S_r$ and across $S_r$ and $V''$. $r$ is chosen such that the total number of queries is at most $n\log{n}$. Goto step 2.
 \end{enumerate}
  
 \paragraph*{Analysis} By the analysis (Lemma \ref{lemma:mlG'}) of Algorithm \ref{algo:cc-error-noside} the extracted subgraph $S$ will have size $\geq c\log{n}$ iff $G''$ contains a subcluster of original $G$ of size $O(c\log{n})$. Moreover, by Lemma \ref{lemma:vertex}, once $S$ is detected $S$ will be fully grown at the end of that step, that is within the next $c$ rounds. 
 
 Now by the same analysis as in the previous section \ref{sec:analysis-r1}, once we choose $4kc\log{n}$ vertices, thus query $16k^2c^2\log{n}$ edges in $\lceil \frac{16k^2c^2}{n}\rceil$ rounds, then with high probability, each cluster with at least $\frac{n}{2k}$ size will have $c\log{n}$ representatives in $G''$ and will be fully grown. We are then left with at most $\frac{n}{2}$ vertices and can apply the argument recursively. Thus the round complexity is $O(\lceil \frac{16k^2c^2}{n}\rceil\log{n}+kc)$ where the second term comes from using at most $c$ rounds for growing cluster $S$ in step $2$. 
 
 If $kc \ge \frac{n}{4\sqrt{\log{n}}}$, then we pick $n^2$ edges in at most $\frac{n}{\log{n}}$ rounds, and the optimum algorithm has round complexity at least 
 $\Theta(\frac{kc}{\log{n}})=\Theta( \frac{n}{4\sqrt{\log{n}\log{n}}})$. So, we are within a $\sqrt{\log{n}}$ factor of the optimum.
 
If $kc \leq \frac{n}{4\sqrt{\log{n}}}$, but $kc \geq \sqrt{\log{n}}$, then the round complexity of our algorithm is $O(kc\sqrt{\log{n}}+\log{n})=O(kc\sqrt{\log{n}})$, again within a $\sqrt{\log{n}}\log{n}$ factor of the optimum.

If $kc \le \sqrt{\log{n}}$, then in the first round, all the clusters that have size at least $c\sqrt{n}$ will have enough representatives, and will be fully grown at the end of step $2$. After that each cluster will have at most $c\sqrt{n}$ vertices. Hence, a total of at most $kc\sqrt{n}\leq \log{n}$ vertices will remain to be clustered. Thus the total number of rounds required will be $O(kc)$, within $\log{n}$ factor of the optimum.

 Recalling that $c=O(\frac{1}{\lambda^2})=O(\frac{1}{\Delta({p\|(1-p)})})$, we get Theorem \ref{thm:error-parallel}.
 
 \begin{theorem-n}[\ref{thm:error-parallel}]
Faulty Oracle with no Side Information. There exists an algorithm for \cc~with faulty oracle with error probability $\frac{1}{2}-\lambda$ and no side information such that it achieves a round complexity within $\tilde{O}(\sqrt{\log{n}})$ factor of the optimum that recovers $\hat{G}$, ML estimate of $G$ with all $\binom{n}2$ queries with high probability. 
  \end{theorem-n}

  This also gives a new parallel algorithm for correlation clustering over noisy input where in each round $n\log{n}$ work is allowed.